\documentclass[
a4paper,
onecolumn
]{quantumarticle}
\pdfoutput=1

\usepackage[
margin=3cm,
]{geometry}

\usepackage{graphicx}
\usepackage{contextuality_env}
\usepackage{dcolumn}% Align table columns on decimal point
\usepackage{bm}% bold math
\usepackage{centernot}
\usepackage{xfrac}
\usepackage[numbers,sort&compress]{natbib}

\numberwithin{equation}{section}

\makeatletter
\AtBeginDocument{\let\LS@rot\@undefined}
\makeatother

\let\originalmiddle=\middle
\def\middle#1{\mathrel{}\originalmiddle#1\mathrel{}}

\begin{document}

\title{Formalizing contextuality in sequential scenarios}
\author{Kim Vallée}
\affiliation{Sorbonne University, CNRS, LIP6, F-75005 Paris, France}
\email{kim.vallee@lip6.fr}
\orcid{0000-0002-4743-3984}

\author{Damian Markham}
\affiliation{Sorbonne University, CNRS, LIP6, F-75005 Paris, France}
\date{24 November 2025}

\begin{abstract}
	This paper provides a framework for characterizing sequential scenarios, allowing for the identification of contextuality given empirical data, and then provides precise operational interpretations in terms of the possible hidden variable model explanations.
	Sequential scenarios are different in essence from non-local scenarios and standard frameworks for contextuality as each instrument is allowed to change the state as it enters subsequent instruments. 
    Thus, it is necessary to formulate the possible state update in any hidden variable model description.
	Here we explore such hidden variable models for sequential scenarios, and we develop on the notion of no-disturbance: an instrument $A$ does not disturb another instrument $B$ if the statistics of $B$ are independent of whether $A$ was measured or not.
We define non-contextuality inequalities for the sequential scenario, and show that violation implies that the data cannot be explained by a hidden variable model that is both deterministic and not disturbing in this sense.
We further provide a translation from standard contextuality frameworks to ours, providing sequential versions which carry over the same inequalities and measures of contextuality, but now with the sequential interpretations stated.

\end{abstract}

\maketitle

\clearpage

\setcounter{tocdepth}{2}
\tableofcontents

\clearpage

\section{Introduction}

Quantum theory and quantum information have many applications, from communication~\cite{spekkens2009PreparationContextualityPowers} to computation speed-up~\cite{howard2014ContextualitySuppliesMagic,raussendorf2013ContextualityMeasurementbasedQuantum}. However, developing these applications requires us to explore more fundamental features of quantum mechanics such as non-locality or contextuality. We will focus on the latter, contextuality, which takes its root in the Kochen-Specker theorem~\cite{kochen1967ProblemHiddenVariables} but has seen renewed interest in recent years~\cite{spekkens2005ContextualityPreparationsTransformations,abramsky2011SheaftheoreticStructureNonlocality,dzhafarov2016ContextualitybyDefaultBriefOverview,acin2015CombinatorialApproachNonlocality,cabello2014GraphTheoreticApproachQuantum}.
In order to validate or invalidate contextuality as a true property of nature, many experiments have been made~\cite{budroni2022KochenSpeckerContextuality}, some of them actually using sequential scenarios, i.e. measuring the system multiple times in a row without destroying it~\cite{lapkiewicz2011ExperimentalNonclassicalityIndivisible,ahrens2013TwoFundamentalExperimental,chen2023ExperimentalDemonstrationQuantum,zhang2013StateIndependentExperimentalTest}. 
At the same time, formal frameworks for contextuality have been powerful tools for understanding its role and application across quantum information, from quantum computation \cite{raussendorf2013ContextualityMeasurementbasedQuantum,howard2014ContextualitySuppliesMagic}, to quantum metrology \cite{lostaglio2020CertifyingQuantumSignatures}, quantum communication complexity \cite{gupta2023QuantumContextualityProvides} and many aspects of quantum cryptography \cite{horodecki2010ContextualityOffersDeviceindependent,spekkens2009PreparationContextualityPowers}. Nonetheless, most of these formal frameworks work with collections of measurements which can be measured together, without specifying how, but they are not able to capture the structure imposed by sequential scenarios explicitly. 
Extending these frameworks, and their interpretations, to sequential scenarios opens the door to understanding and demonstrating these non-classical effects in many more situations.

There are several works which explore the relationship between contextuality and sequential scenarios~\cite{guhne2010CompatibilityNoncontextualitySequential,budroni2019ContextualityMemoryCost,abramsky2023CombiningContextualityCausality,searle2024CorrelationCombinatoricsCausal}, all with different motivations. In particular, a recent line of research considered finite memory in sequential scenarios~\cite{hoffmann2018StructureTemporalCorrelations,kleinmann2011MemoryCostQuantum,budroni2019MemoryCostTemporal,budroni2019ContextualityMemoryCost,searle2024CorrelationCombinatoricsCausal}, allowing one to state the memory requirement of any system to perform some correlations.
There also have been some progress relating non-local scenarios to prepare and measure scenarios~\cite{wright2023InvertibleMapBell,catani2024ConnectingXORXOR,baroni2024QuantumBoundsCompiled}, and on prepare and measure scenario with one transformation~\cite{pusey2015LogicalPrePostselection,leifer2005LogicalPrePostSelection}.
Our paper complements the existing approaches, which either lack a clear hidden variable model, cannot reproduce all the scenarios we allow in our framework or have different motivations. In particular, we focus on a notion of contextuality derived from the sheaf theoretic framework for contextuality, 
which we then relate to the existence of hidden variable model explanations with certain constraints.
In order to do this, we propose a general framework for sequential scenarios, and their interpretations, which builds on the ideas of the General Probabilistic Theory (GPT) and Operational Probabilistic Theory (OPT) formalisms~\cite{plavala2023GeneralProbabilisticTheories,heinosaari2019NofreeinformationPrincipleGeneral,dariano2017QuantumTheoryFirst,dariano2022IncompatibilityObservablesChannels}. 

At the same time we must capture in our framework the idea of when measurements, and the underlying instruments, can be understood to be measurable at the same time, or independent of order. 
One common way of doing this is to define instrument compatibility. 
However, we will not use this notion, for two reasons, first instrument compatibility has many definitions~\cite{buscemi2023UnifyingDifferentNotions}, and it is not clear which definition should be used in which case. Secondly, it is relatively demanding to check operationally~\cite[Section 2.8.4]{dariano2017QuantumTheoryFirst}, as it would require a tomographically complete set of preparations and measurements, in addition to the instruments.
Instead, we propose another notion called instrument no-disturbance inspired from~\cite{heinosaari2010NondisturbingQuantumMeasurements}, which is also operational and only requires a tomographically complete set of preparations to be checked.
We note that no-disturbance -- as an idea -- has been studied a lot and can be related to compatibility in some approaches (e.g. in OPT~\cite{dariano2022IncompatibilityObservablesChannels}). Our proposed formulation is motivated by its simplicity, and as far as we know, it has not been formalized before in this way for general hidden variable models.

Combining the two points above, we give a formal framework to study contextuality in sequential scenario with operational assumptions, and we go beyond by showing how to relate measurements scenarios~\cite{abramsky2011SheaftheoreticStructureNonlocality} to the notion of sequential scenarios we introduce.

\bigskip
The structure of the paper is as follows. We start by defining what sequential scenarios are in Section~\ref{sec:sequential-scenario}, and what is a hidden variable model (HVM) in such sequential scenarios. As an example, we make sure to retrieve quantum theory as a special case of HVMs in sequential scenarios.
We follow up by introducing operational considerations and HVM restrictions in Section~\ref{sec:operational-restrictions}, looking at outcome determinism, no-disturbance and outcome independence. 
We introduce the formal notion of non-contextuality for sequential scenarios in Section~\ref{sec:non-contextuality}, showing that the non-contextual set is a convex polytope and that the contextual fraction holds. 
Here contextuality refers to the property that statistics of individual `contexts' or, in our case, instrument sequences, can be understood as coming form marginals of some global probability distribution.
The main point of our paper is reached in Section~\ref{sec:restriced_HVMs_and_non-contextuality}, where we make a formal link between operational restrictions and contextuality. 
We show in Theorem~\ref{thm:ND-and-OD-are-NC} that contextuality, as formally defined above, is equivalent to the existence of an HVM that is both non-disturbing, and outcome deterministic. We also show that for scenarios where sequences are at most length two non-contextuality can also be shown to be equivalent to there being an HVM that is non-disturbing and outcome independent.
Finally, we propose a mapping between measurement scenarios as described by~\cite{abramsky2011SheaftheoreticStructureNonlocality} and sequential scenarios in Section~\ref{sec:meas-scen-to-seq}.

\section{Sequential scenarios} \label{sec:sequential-scenario}

In order to talk about features of a sequential scenario, it is necessary to give a precise description of what we consider to be sequential scenario. We will roughly follow the steps in the Sheaf theoretic framework for contextuality~\cite{abramsky2011SheaftheoreticStructureNonlocality} where the same is done for measurement scenarios. First we introduce the notion of a sequential scenario followed by empirical behaviours on sequential scenarios. Then, with these building blocks, we will have the necessary material to talk of hidden variable models and, finally, of contextuality.

\subsection{Definitions}

In sequential scenarios, we are working with \textit{instruments}:
\begin{definition}[Instrument]
    An instrument $A\colon \Sigma_A \to \Sigma_A' \times O_A$ takes as input a system in state a $\sigma_A \in \Sigma_A$ and outputs a system in a state $\sigma_A' \in \Sigma_A'$ together with a classical outcome $a \in O_A$. Where $\Sigma_A$ and $\Sigma_A'$ are state spaces and $O_A$ is the set of classical outcomes allowed by the instrument.
\end{definition}
Instruments have had some interest in the literature, either in quantum theory~\cite{buscemi2023UnifyingDifferentNotions,heinosaari2010NondisturbingQuantumMeasurements} or in GPTs~\cite{plavala2023GeneralProbabilisticTheories} and OPTs~\cite{dariano2022IncompatibilityObservablesChannels} -- therein called tests.

For some instrument $B$, it might be that its outcome statistics remain unchanged whether $A$ is placed before or not. In that case we say that $A$ does not disturb $B$:
\begin{definition}[Instrument no-disturbance] \label{def:instrument_no-disturbance}
	Let two instruments $A: \Sigma \to \Sigma \times O_{A}$ and $B: \Sigma \to \Sigma' \times O_{B}$. We say that $A$ does not disturb $B$, written $A \nd B$, if for every outcome $b \in O_{B}$ and every state $\sigma \in \Sigma$ we have:
    \begin{equation} \label{eq:instrument-no-disturbance}
        p_{B}(b\vert \sigma) = \sum_{a, \sigma'} p_{A}(a \vert \sigma) p_{A}(\tilde{\sigma} \vert a, \sigma) p_{B}(b\vert \tilde{\sigma})
    \end{equation}
    where $p_{A}(a \vert \sigma)$ (respectively $p_{B}(b \vert \sigma)$) is the probability to have outcome $a$ ($b$) when $A$ ($B$) is performed, and the input state is $\sigma$. Similarly, $p_{A}(\tilde{\sigma} \vert a, \sigma)$ is the probability for instrument $A$ to prepare the state $\tilde{\sigma}$ given the input state $\sigma$ and the outcome $a$.
\end{definition}

An immediate remark about Equation~\eqref{eq:instrument-no-disturbance} is that one could write $p_{A}(a \vert \sigma) p_{A}(\tilde{\sigma} \vert a, \sigma)$ as $p_{A}(\tilde{\sigma}, a\vert \sigma)$ by using the chain rule. Throughout this paper we will keep the first notation, with two terms, in order to dissociate the process of generating a new state and the process of generating an output.

In order to test instrument no-disturbance, one requires \textit{only} a tomographically complete set of preparation, since it should hold for all state. It should be contrasted with compatibility, which requires \textit{at least} a tomographically complete set of measurements and states according to~\cite[Section 2.8.4]{dariano2017QuantumTheoryFirst}.

Definition~\ref{def:instrument_no-disturbance} deserves some additional comments regarding the existing literature. 
First, the term no-disturbance, or the \enquote{Gleason property}, has sometimes been used in the literature~\cite{cabello2010NonContextualityPhysicalTheories,ramanathan2012GeneralizedMonogamyContextual} to describe the compatibility of marginals, i.e.:
\begin{equation}
	\sum_{b} p(a,b) = \sum_{c} p(a,c) = p(a)
\end{equation}
It is not equivalent to our definition, and does not carry the same insights. However, these two notions are related as we mention in Section~\ref{ssec:no-disturbance-and-com}.
	On the other hand, No Signalling In Time (NSIT) from~\cite{kofler2013ConditionMacroscopicRealism} is similar to our notion of no-disturbance, but it is applied differently; NSIT should apply unconditionally to all instruments, whereas no-disturbance is assumed to hold for only some set of instruments.
Finally, we do not refer here to the notion of state no-disturbance, as in the \enquote{no information without disturbance} principle~\cite{busch2009NoInformationDisturbance,heinosaari2016InvitationQuantumIncompatibility}, but rather to outcome no disturbance -- i.e. the fact that the outcome of an instrument is independent of the previous instrument.

When one has multiple sequences of instruments, this forms a \textit{sequential scenario}:
\begin{definition}[Sequential scenario] \label{def:seq-scenario}
    A sequential scenario is given by a tuple $\XMOSEQ$ where:
    \begin{itemize}
        \item $X$ is the set of instrument labels;
        \item $\mathcal{S}$ the set of sequences, where a sequence $S \in \mathcal{S}$ is a set of instrument labels indexed by their ordering $S = \{ A_i \mid A \in X, i \in I\}$ where $I \subset \mathbb{N}$ is an indexed set. The set of all possible sequences given a set $X$ and an indexed set $I$ is written $\mathfrak{S}_I^X$; % such that $S \in X^+$ where $X^+ = \bigcup_{n > 0} X^n$ and $X^n$ is the set of sequences of elements from $X$ with length $n$
        \item $O = (O_x)_{x \in X}$ is the set of classical outcomes for each instrument. When we consider a sequence of instruments $S \in \mathfrak{S}_I^X$, we define $O_S$ as the Cartesian product of the possible outcomes:
	\begin{equation}
            O_S = \prod_{ \substack{s \in S, x \in X \\ i \in I \\ x_i = s}} O_x
	.\end{equation} 
        Where $I$ is an indexed set and $S$ an arbitrary sequence.
    \end{itemize}
\end{definition}
%For simplicity, we introduce the notation $A_i \sim A$, which means that two instruments are the same up to the index.

We point out two implicit assumptions which are made in sequential scenarios. First, all the first instruments in a sequence must have the same input state space, say $\Sigma_{\textup{in}}$, and second whenever two instruments are made in a sequence they must have matching state space. As an example, let $A\colon \Sigma_A \to \Sigma_A' \times O_A$ and $B\colon \Sigma_B \to \Sigma_B' \times O_B$ such that $S = \set{A_1, B_2}$ then one must have $\Sigma_A' = \Sigma_B$. Moreover, since $A_1$ is the first instrument in the sequence we necessarily have $\Sigma_A = \Sigma_{\textup{in}}$.
% Note that sequences are straightforward generalization of contexts in the Sheaf theoretic formulation of contextuality, where we allow instruments to be repeated and where we add an index to each instrument.

We give below an example with a sequential version of the KCBS scenario~\cite{klyachko2008SimpleTestHidden,dourdent2018ContextualityWitnessQuantum}.

\newcommand{\KCBS}{\textup{KCBS}}

\begin{example}[Sequential KCBS] \label{ex:seq-KCBS}
    We define the sequential KCBS scenario $\XMOSEQ[\textup{KCBS}]$ where we have five instruments $X_{\KCBS} = \left\{A,B,C,D,E\right\}$, and they all have two possible outcomes $\forall I \in X_{\KCBS} \colon (O_{\KCBS})_{I} = \{0,1\}$. Then we define the sequences as:
    \begin{equation}
        \begin{split}
            \mathcal{S}_{\KCBS} = \big \{ &\{A_1,B_2\}, \{B_1,C_2\}, \\
            &\{C_1, D_2\}, \{D_1, E_2\},\\
            &\{E_1, A_2\}\big \}
        \end{split}
    \end{equation}
    Where the subscript is the order of the instrument inside a sequence.
\end{example}

Note that in Example~\ref{ex:seq-KCBS} we could have decided of different order in the sequences $S \in \mathcal{S}_{\KCBS}$, and that would still be a valid scenario. This illustrates that there are multiple possible sequential scenarios associated to a measurement scenario. We will give more details on this possibility in Section~\ref{sec:non-contextuality}.

When we run an experiment or make a model from a theory of a given sequential scenario, we obtain a set of probabilities. We call these probabilities an \textit{empirical behaviour}, which is defined on a sequential scenario:
\begin{definition}[(Sequential) Empirical behaviour] \label{def:empirical-sequential}
    An empirical behaviour $e$ on a sequential scenario $\XMOSEQ$ is defined as a family of probability distribution $e = (e_S)_{S\in \mathcal{S}}$. Each $e_S\colon O_S \to [0,1]$ is a probability distribution on classical outcomes for the instruments in a sequence.
\end{definition}

We also define the notion of compatibility of marginals, in the same spirit as~\cite{abramsky2011SheaftheoreticStructureNonlocality}:
\begin{definition}[Compatibility of marginals] \label{def:compatibility-of-marginals}
    We say that an empirical behaviour $e$ on a sequential scenario $\XMOSEQ$ respects the compatibility of marginals if for any instrument $A\in X$, that appears in two sequences $S_1, S_2 \in \mathcal{S}$ as $A_k$ in $S_1$ and as $A_{k'}$ in $S_2$, where $k, k' \in \mathbb{N}$ then:
    \begin{equation}
        \sum_{a_i \neq a_k} p(a_1, \dots, a_{\norm*{S_1}}\vert S_1) = \sum_{a_j \neq a_{k'}} p(a_1, \dots, a_{\norm*{S_2}}\vert S_2)
    .\end{equation}
    Where the outcomes $a_k$ and $a_{k'}$ correspond to the instruments $A_k$ and $A_{k'}$ respectively.
    %$i,j \in \mathbb{N}^+$, $i \leq \norm*{S_1}$, $j \leq \norm*{S_2}$ and it holds for all 
\end{definition}
Notably, the compatibility of marginals matches the notion of no-signalling in the case of non-local scenarios, such as Bell scenarios~\cite{abramsky2011SheaftheoreticStructureNonlocality}. 

\subsection{Hidden variable model} \label{sec:operation-theory}

Since this paper aims at deriving an operational way to rule out some theories, it is useful to propose a model which can reproduce the statistics of any theory. Such a model is called a Hidden Variable Model (HVM), and it can reproduce classical theory, quantum theory or any post-quantum theory.
Nevertheless, restricting a HVM may result in a specific theory, e.g., to reproduce non-contextual theories the Kochen-Specker theorem constrains a HVM to be deterministic and parameter independent~\cite{kochen1967ProblemHiddenVariables}.

In this section, we provide all the building blocks of a HVM for sequential scenarios, which we then combine\footnote{In this paper we use the terminology HVM and not ontological model to match the historical notation in the sheaf theoretic framework for contextuality, but they should be understood as conceptually the same.}.
In the following, we will refer to a hidden variable as $\lambda$ and to the set of hidden variables as $\Lambda$, i.e. $\lambda \in \Lambda$. We consider $\Lambda$ to be finite, but arbitrarily large, to simplify notation without losing generality.
In the HVM, instruments will be represented by devices which take as input a hidden variable and output a new hidden variable together with an outcome $A \colon \Lambda_A \to \Lambda_A' \times O_A$, where $\Lambda_A \subseteq \Lambda$ and $\Lambda_A' \subseteq \Lambda$. For simplicity of notation we will always write $A \colon \Lambda \to \Lambda \times O_A$ for any instrument.

An arbitrary sequence of instruments is depicted in Figure~\ref{fig:measurement-dissection}, where first a hidden variable $\lambda$ is prepared following a distribution $\mu_{\rho}\left( \lambda \right) $ given by a state preparation $\rho$. Then the hidden variable goes multiple times through instruments, producing an outcome and updating for each instrument\footnote{One might note that all these notions are special cases of channels~\cite{plavala2023GeneralProbabilisticTheories}, but we keep them separated voluntarily in order to simplify the conditions we impose on them.}. We formalize this procedure in this section by preparation functions, response functions and transfer functions respectively.

\begin{figure}[ht!]
    \centering
    \includegraphics[width=0.8\textwidth]{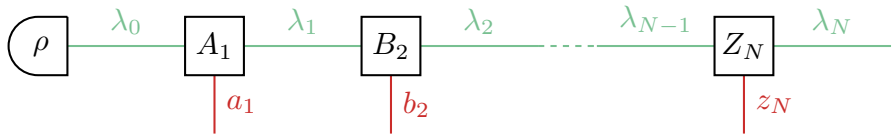}
    \caption{Example of a sequence of size $N$, where the hidden variables are represented in green, and the classical outcomes in red. First there is a state preparation $\rho$ where the hidden variable $\lambda_0$ is sampled with probability $\mu_\rho(\lambda_0)$. This hidden variable is then an input to the instrument $A_1$. It has two results: one is a classical outcome $a_1$ given with probability $\xi_{A_1}(a_1\mid \lambda_0)$, and the second a new hidden variable $\lambda_1$. This new hidden variable is sampled from the whole space $\Lambda$ with probability $\Gamma_{A_1}(\lambda_1\mid \lambda_0, a_1)$. The new hidden variable produced is forwarded to the next instrument $B_2$, and it produces a classical outcome $b_2$ as well as a hidden variable $\lambda_2$. This continues until the end of the sequence.}
    \label{fig:measurement-dissection}
\end{figure}

\subsubsection{Preparation function} \label{ssec:preparation-function}

A state preparation $\rho$ samples hidden variables with an associated probability distribution $\mu_{\rho}\colon \Lambda \to [0,1]$ which has the properties:
\begin{subequations}
	\begin{align}
	    \forall \lambda \in \Lambda\colon \mu_{\rho}(\lambda) &\geq 0 \\
	    \sum_{\lambda\in \Lambda} \mu_{\rho}(\lambda) &= 1\,.
	\end{align}
\end{subequations}
The support of this probability distribution is written $\Lambda_\rho = \{\lambda \vert \mu_\rho(\lambda) > 0\}$. We call $\mu_\rho$ the \textit{preparation function} of $\rho$.

\subsubsection{Response function} \label{ssec:response-function}

Each instrument $A$ responds to a hidden variable by outputting a classical outcome, which we denote by the lower case letter $a$, and we model it by a \textit{response function} $\xiid{A}{\cdot}{\lambda}\colon O_{A} \to [0,1]$, which is the probability to have the outcome $a$ given the hidden variable $\lambda$ when measuring $A$. These \textit{response functions} are normalized -- on the outcomes -- and positive such that for all $\lambda \in \Lambda$:
\begin{subequations}
	\begin{alignat}{3}
	    \label{eq:normalization-response-functions} &\sum_a &\xiid{A}{a}{\lambda} &= 1 \\
	    \label{eq:positive-response-functions} \forall a \in O_A\colon& &\xiid{A}{a}{\lambda} &\geq 0\,.
	\end{alignat}
\end{subequations}

Notice that each response function is independent of the order of the instrument, this information may only be accessible through the hidden variable $\lambda$. In other words we make the following assumption:
\begin{equation} 
	\forall k \in \mathbb{N}\colon \quad \xi_{A_k}(\cdot | \lambda) \equiv \xi_{A}(\cdot | \lambda)
.\end{equation} 

\subsubsection{Transfer function} \label{ssec:transfer-function}

After a given instrument $A_i$, at step $i$, the previous hidden variable $\lambda_{i-1}$ is resampled in a new hidden variable $\lambda_{i}$, possibly stochastically. The resampling of the hidden variable is determined following the function $\Gammaid{A_i}{\cdot}{\lambda_{i-1}, a_i}\colon \Lambda \to [0,1]$, which represents the probability to sample $\lambda_i$ given the previous hidden variable $\lambda_{i-1}$ and the classical outcome $a_i$. We naturally assume that the resampling is normalized, and it behaves as a probability distribution such that $\forall \lambda_{i-1} \in \Lambda$ and for all $a_i \in O_{A_i}$:
\begin{subequations}
	\begin{align}
	    \label{eq:normalized-resampling} \sum_{\lambda_i \in \Lambda} &\Gammaid{A_i}{\lambda_i}{\lambda_{i-1}, a_i} = 1 \\
	    \label{eq:resampling-positive} &\Gammaid{A_i}{\lambda_i}{\lambda_{i-1}, a_i} \geq 0\,.
	\end{align}
\end{subequations}

As one may notice, we have decided to remove the dependency in all the past instruments from the transfer function, e.g.:
\begin{equation}
	\Gamma_{A_i}(\lambda_i\vert \lambda_{i-1}, \dots, \lambda_0, a_i, \dots, a_0)
.\end{equation} 
However, we do not lose the generality, since this can be encoded in the state $\lambda_{i-1}$ as it may have arbitrarily large dimension.

\subsubsection{No-backward-in-time signalling} \label{sssec:no_backward_in_time_signalling}

The most general HVM cannot be built only from preparation, transfer and response functions, only because one could allow for signalling backward in time, which is highly counter-intuitive. Consequently, in the remaining of this paper, we make the assumption of No-Backwards-in-Time-Signalling (NBTS), similar to the one described in~\cite{guryanova2019ExploringLimitsNo}. This notion is also sometimes referred to as the Arrow of Time (AoT) constraint~\cite{budroni2019ContextualityMemoryCost,clemente2016NoFineTheorem}.

\begin{definition}[No-Backwards-in-Time-Signalling] \label{def:no-backwards-in-time-signalling}
	Any outcome or state at time $t_0$ cannot depend on any outcome or state that occurs at a later time  $t_1 > t_0$. As a consequence, any probability distribution should marginalize on future events.
\end{definition}

We have not formally introduced the notion of time in our framework, but it is implicitly assumed from the ordering of the instruments.
For instance in a sequential scenario $\XMOSEQ$, let a sequence $S \in \mathcal{S}$ of length $N$. Such that the first instrument is $A_1$ and the last is $Z_N$. Take the outcome that occurs at step $i$, then by the NBTS condition:
\begin{equation} \label{eq:NBTS}
	\begin{split}
		p&\left( c_i  \mid a_1, \ldots, z_N, \lambda_0, \ldots, \lambda_N \right) \\
	&= p\left( c_i  \mid a_1, \ldots, b_{i-1}, \lambda_0, \ldots, \lambda_{i-1} \right).
	\end{split}
\end{equation}

\subsubsection{Hidden variable model}

Lastly, we give the definition of a hidden variable model with the NBTS condition:

\begin{definition}[NBTS sequential hidden variable model] \label{def:HVM}
    A hidden variable model on a sequential scenario $\XMOSEQ$ is a tuple $\HVM$ such that:
    \begin{enumerate}[label=(\roman*)]
        \item $\Lambda$ is the set of hidden variables;
        \item $\mu \colon \Lambda \to [0,1]$ is the preparation function;
        \item for all $\lambda \in \Lambda$, $h^\lambda = (h^\lambda_S)_{S \in \mathcal{S}}$ is a family of probability distributions  $h^\lambda_S\colon O_S \to [0,1]$ such that for all $o \in O_S$:
		\begin{equation} \label{eq:NSFF-empirical-model}
		    h^\lambda_S(o) = \prod_{i=1}^{\left\lvert S \right\rvert} \xi_{A_i}(a_i \mid \lambda_{i-1}) \Gamma_{A_i}(\lambda_i \mid \lambda_{i-1}, a_i)
		.\end{equation}
		Where $o \coloneq (a_1, \ldots, a_{\left\lvert S \right\rvert})$.
    \end{enumerate}
    These give rise to an empirical model:
    \begin{equation} \label{eq:generic-HV-behaviour}
        h \coloneq \sum_{\lambda \in \Lambda}\mu(\lambda) h^\lambda
    \end{equation}
\end{definition}

In the following, we will always consider NBTS sequential hidden variable models.

\subsection{Quantum theory as a special case} \label{ssec:Quantum-example-HVM}

It is important to check that our HVM can at least recover quantum theory as a special case. 
This is possible if one assigns the set of hidden variables to the set of pure states $\Lambda = \mathcal{H}$, where $\mathcal{H}$ is a complex Hilbert space of arbitrary dimension. This means that the hidden variables are in one to one correspondence to pure states $\lambda \mapsto \ket{\lambda}$.

With this in mind, we can define the response functions and the transfer functions quantum mechanically for any instrument $A$ and outcome $a \in O_A$: 

\begin{align}
	\label{eq:HVM-to-quantum-xi} \xiid{A}{a}{\lambda_0} &\coloneq \Tr\left [ \mathcal{I}_A^a\left( \ketbra{\lambda_0}{\lambda_0} \right)  \right ] \\
    \label{eq:HVM-to-quantum-rho} \rho_{a} &\coloneq \frac{\mathcal{I}_A^{a}(\ketbra{\lambda_0}{\lambda_0})}{\xiid{A}{a}{\lambda_0}} \\
    \label{eq:HVM-to-quantum-gamma} \Gammaid{A}{\lambda_1}{\lambda_0, a} &\coloneq \Tr\left[ \rho_{a} \ketbra{\lambda_1}{\lambda_1}  \right]
\end{align}
where $\mathcal{I}_A^{a}$ here holds for the quantum instrument of $A$ given outcome $a$, where we use the definition of~\cite{heinosaari2010NondisturbingQuantumMeasurements} and~\cite[Chapter 4]{davies1976QuantumTheoryOpen}. A quantum instrument is defined as a set of Completely Positive Linear Map $\mathcal{I}_A \equiv \{\mathcal{I}_A^a\}$ and which satisfy $\mathcal{I}_A^{*a}(\mathbbm{1}) = M_{A}^a$ for every $a \in O_A$ where $\mathcal{I}_A^{*a}$ is the dual map and $M_{A}^a$ is quantum POVM element.
The above equations can be generalized to mixed states knowing that the instruments are linear operators.

\section{Operational considerations and HVM restrictions} \label{sec:operational-restrictions}

Section~\ref{ssec:Quantum-example-HVM} shows that a NBTS HVM can indeed reproduce all the effects of quantum theory. Even so, one might consider restricted HVMs, that aligns with some additional physical principles. In the case of the Bell theorem, we also impose any HVM to respect the locality condition, showing that nature is non-local. Similarly, here we can go further than NBTS, and force more constraints upon the HVM.

Of course, imposing supplementary constraints must be motivated from some considerations and here we will study operational considerations. An operational notion is a notion that we can observe directly, something tied to experimental results, and as theory independent as possible\footnote{Note that any operational theory is still necessarily theory-laden due to the Duhem-Quine thesis~\cite{spekkens2025QuantumFoundations}.}.
An example of an operational approach is the one Spekkens proposed for his notion of contextuality~\cite{spekkens2005ContextualityPreparationsTransformations} where two preparations are equivalent if no measurement can distinguish them. This means that any experiment with a tomographically complete set of measurements should observe no difference in the statistics of these two equivalent preparations.
This operational fact is then reflected on the HVM, by restricting the preparation functions to be equal for all hidden variable, therein called \textit{preparation non-contextuality}.

In this section, we introduce three HVM restrictions: outcome determinism, no-disturbance and outcome independence. We will also relate them to experimental considerations, to show how they can be operationally defined.
Together, some of these restrictions are related to non-contextuality, as detailed in Section~\ref{sec:non-contextuality}.

The behaviour of an instrument $A$ in a HVM will be described by its response function $\xi_{A}$ and its transfer function $\Gamma_{A}$ as described by Section~\ref{sec:operation-theory}.

\subsection{Outcome determinism}

The notion of outcome determinism for sequential scenarios is directly derived from the one used in the Kochen-Specker theorem or in Fine's theorem.

\begin{definition}[Outcome determinism] \label{def:outcome-determinism}
    In a HVM, an instrument $A$ is called outcome deterministic if the probabilities associated with its classical outcome are always $0$ or $1$. More formally, for any hidden variable  $\lambda \in \Lambda$, we say that $A$ is outcome deterministic if:
    \begin{equation}
	    \xi_{A}(a \mid \lambda) \in \{0,1\}
    .\end{equation}
\end{definition}

In order to justify determinism, we can consider projectors, or \textit{sharp} instruments~\cite{chiribella2014MeasurementSharpnessCuts,kochen1967ProblemHiddenVariables}, preparation non-contextuality as illustrated in~\cite{spekkens2005ContextualityPreparationsTransformations}, or believe that the theory at play is maximally $\psi$-epistemic~\cite{leifer2013MaximallyEpistemicInterpretations}.

\subsection{No-disturbance}

An instrument $A$ does not disturb another instrument $B$ if placing $A$ right before $B$ does not change the probability of outcomes of $B$. As an example, a fair coin flip is never disturbed by any previous instrument, because it will always outcome either head or tails with half probability. Similarly, a non-invasive instrument will never disturb any following instrument, because since the system is not disturbed the following instrument must act on the undisturbed system. Any repeatable outcome deterministic instrument is also self non-disturbing. We analyse these examples after introducing the formal definition of no-disturbance for a HVM:
\begin{definition}[No-disturbance] \label{def:no-disturbance}
    In a HVM, an instrument $A$ does not disturb another instrument $B$, which we write $A \nd B$, if the following holds:
    \begin{equation}
    \label{eq:no-disturbance}
    \begin{split}
        &\forall \lambda \in \Lambda\colon \forall b \in O_{B}\colon\\
        &\sum_{a, \lambda'} \xiid{A}{a}{\lambda} \Gammaid{A}{\lambda'}{\lambda, a}\xiid{B}{b}{\lambda'} = \xiid{B}{b}{\lambda}
    \end{split}
    \end{equation}
    Meaning that the outcome statistics remain unchanged for $B$ after $A$ has been measured.
\end{definition}
Definition~\ref{def:no-disturbance} is equivalent to Definition~\ref{def:instrument_no-disturbance} but reformulated for HVMs. Furthermore, it should be noted that Definition~\ref{def:no-disturbance} is not necessarily symmetric, $B$ might not be disturbed by $A$, and yet $A$ can be disturbed by $B$.

Verifying Equation~\eqref{eq:no-disturbance} requires one to be able to prepare every hidden variable $\lambda \in \Lambda$, which operationally means that one should have a tomographically complete set of preparations. We point out that having a tomographically complete set of preparations is commonly discussed in operational contextuality~\cite{mazurek2016ExperimentalTestNoncontextuality,pusey2019ContextualityAccessTomographically,gitton2022SolvableCriterionContextuality} and in the more general settings of GPTs~\cite{schmid2021CharacterizationNoncontextualityFramework}.

\subsubsection{Example: Fair coin flip}

The fair coin flip instrument $A_{\textup{fcf}}$ is never disturbed by any previous instrument. In order to prove it let us use Definition~\ref{def:no-disturbance}, where we have the fair coin flip response function to be $\xiid{A_{\textup{fcf}}}{0}{\lambda} = \xiid{A_{\textup{fcf}}}{1}{\lambda} = 0.5$ for all hidden variable $\lambda \in \Lambda$. Then:
\begin{align*}
    &\sum_{a, \lambda'} \xiid{A}{a}{\lambda} \Gammaid{A}{\lambda'}{\lambda, a}\xiid{A_{\textup{fcf}}}{a_{\textup{fcf}}}{\lambda'} \\
    &\quad = \frac{1}{2} \sum_{a, \lambda'} \xiid{A}{a}{\lambda} \Gammaid{A}{\lambda'}{\lambda, a} \\
    &\quad = \frac{1}{2}
\end{align*}
where this holds for any instrument $A$ and any$\lambda\in\Lambda$. Thus, the fair coin flip is never disturbed by any previous instrument, however it might disturb further instruments, depending on the specific state transformation that it implements.

\subsubsection{Example: Non-invasive instrument}

A non-invasive instrument $A_{\textup{NI}}$, is a device that does not change the hidden variable, i.e., $\Gamma_{A_{\textup{NI}}}(\lambda'\vert\lambda) = \delta_{\lambda,\lambda'}$. In that case, for any following instrument $A$ Definition~\ref{def:no-disturbance} is satisfied:
\begin{align*}
    &\sum_{a_{\textup{NI}}, \lambda'} \xiid{A_{\textup{NI}}}{a_{\textup{NI}}}{\lambda} \Gammaid{A_{\textup{NI}}}{\lambda'}{\lambda, a}\xiid{A}{a}{\lambda'} \\
    &\quad = \sum_{a_{\textup{NI}}, \lambda'} \xiid{A_{\textup{NI}}}{a_{\textup{NI}}}{\lambda} \delta_{\lambda,\lambda'}\xiid{A}{a}{\lambda'} \\
    &\quad = \xiid{A}{a}{\lambda}
\end{align*}

Thus, it is quite clear that a non-invasive instrument is always non-disturbing according to our definition.

\subsubsection{Example: Repeatable instrument}

\newcommand{\Arep}{A_{\textup{r}}}
\newcommand{\arep}{a_{\textup{r}}}

We define a repeatable instrument $\Arep$ as a device which, once put in sequence, always output the same classical outcome. Mathematically it means that for any hidden variable $\lambda_0 \in \Lambda$ and any classical outcomes $\arep, \arep' \in O_{A_{\textup{rep}}}$ we must have the following:
\begin{align}
    \label{eq:repeatable-measurement} \sum_{\lambda_1} &\xi_{\Arep}(\arep \mid \lambda_0) \Gamma_{\Arep}(\lambda_1 \mid \lambda_0, \arep) \xi_{\Arep}(\arep' \mid \lambda_1) = \xi_{\Arep}(\arep \mid \lambda_0) \delta_{\arep,\arep'}
\end{align}

Such an instrument is in fact self-non-disturbing $\Arep \nd \Arep$. We can show it by summing on both sides of Equation~\eqref{eq:repeatable-measurement}:

\begin{subequations}
	\begin{alignat}{2}
		&\sum_{\lambda_1} \xi_{\Arep}(\arep \mid \lambda_0) \Gamma_{\Arep}(\lambda_1 \mid \lambda_0, \arep) \xi_{\Arep}(\arep' \mid \lambda_1) &&= \xi_{\Arep}(\arep \mid \lambda_0) \delta_{\arep,\arep'} \\
				\implies& \sum_{\lambda_1, \arep} \xi_{\Arep}(\arep \mid \lambda_0) \Gamma_{\Arep}(\lambda_1 \mid \lambda_0, \arep) \xi_{\Arep}(\arep' \mid \lambda_1) &&= \sum_{\arep} \xi_{\Arep}(\arep \mid \lambda_0) \delta_{\arep,\arep'} \\
					& &&= \xi_{\Arep}(\arep' \mid \lambda_0)
	\end{alignat}
\end{subequations}

\subsubsection{No-disturbance in quantum theory}

The notion of no-disturbance is defined for POVM measurements in quantum theory in~\cite{heinosaari2010NondisturbingQuantumMeasurements}. We provide an adapted definition for instruments below: 
\begin{definition}[(Quantum) Non-disturbing instruments from~\cite{heinosaari2010NondisturbingQuantumMeasurements}] \label{def:quantum_non-disturbing_measurements_from--cite-heinosaari2010nondisturbing}
	Given two instruments $A$ and $B$, $A$ can be performed without disturbing $B$, if:
	\begin{equation} \label{eq:no_disturbance_quantum_theory}
		\begin{split}
		    &\forall \rho \in \mathcal{T}\left( \mathcal{H} \right)  \colon \forall b \in O_{B}\colon\\
		    &\quad \Tr \left [ \mathcal{I}_{B}^{b}\left( \sum_{a} \mathcal{I}_A^{a}(\rho) \right) \right ] = \Tr\left [ \mathcal{I}_{B}^{b}\left( \rho \right)  \right ]
		\end{split}
	\end{equation}
	Where $\mathcal{H}$ is a Hilbert space of arbitrary dimension and $\mathcal{T}\left( \mathcal{H} \right)$ is the set of positive trace class operators of trace one.
\end{definition}
We observe that HVMs adapted to quantum theory as shown by Section~\ref{ssec:Quantum-example-HVM} (Equations~\eqref{eq:HVM-to-quantum-xi}--\eqref{eq:HVM-to-quantum-gamma}) in fact retrieve Equation~\eqref{eq:no_disturbance_quantum_theory} when applied to our definition of no-disturbance (see Definition~\ref{def:no-disturbance}).

\subsection{Outcome independence}

We also introduce another restriction called \textit{outcome independence} (OI).
An instrument $A_i$ in a sequence $S$ is outcome independent if its transfer function $\Gamma_{A_i}$ only depends on the previous hidden variable $\lambda_{i-1}$ and not on the outcome $a_{i}$. Namely, the new hidden variable $\lambda_i$ after the instrument must be independent of the outcome $a_i$.

A typical OI instrument samples randomly after the classical outcome production, or the opposite, it deterministically produces the same hidden variable after the classical outcome. We will go more in depth with these examples, but first we provide the definition of outcome independence:

\begin{definition}[(Transfer functions) Outcome independence] \label{def:outcome-independence}
	In a HVM, an instrument $A$ is outcome independent if its transfer function $\Gamma_{A}$ does not depend on the outcome $a$ of the instrument. In other words for all in and out hidden variables $\lambda, \lambda' \in \Lambda$ and for any pair of outcomes $a, a' \in O_{A}$:
    \begin{equation}
            \Gammaid{A}{\lambda'}{\lambda, a} = \Gammaid{A}{\lambda'}{\lambda, a'}
    .\end{equation}
    For simplicity, we will drop the dependency in the outcome when referring to a transfer function that is outcome independent, i.e., $\Gamma_{A}(\lambda' \mid \lambda) \equiv \Gammaid{A}{\lambda'}{\lambda, a}$.
\end{definition}

We would like to highlight the fact that outcome independence is difficult to operationally motivate, particularly since quantum mechanics does have outcome dependence. It is not a requirement to show contextuality, however one can retrieve a notion of contextuality from this assumption, which we make explicit in Theorem~\ref{thm:ND-and-OI-are-NC}.  

\subsubsection{Example : Random resampling}

\newcommand{\rr}{\textup{rr}}

An instrument $A_{\textup{rr}}$ is randomly resampling if given a set of hidden variables $\Lambda$ we have:
\begin{equation}
    \Gamma_{A_\rr}(\lambda'\vert \lambda, a_\rr) = \frac{1}{\norm*{\Lambda}}
\end{equation}

Clearly, for such an instrument, the transfer function does not depend on the classical outcome $a_\rr$.

\subsubsection{Example : Deterministic reset}

\newcommand{\dr}{\textup{dr}}

Similarly to the random resampling, a deterministic reset instrument $A_\dr$, is an instrument that resets the hidden variable $\lambda$ after producing the classical outcome $a_\dr$:
\begin{equation}
    \Gamma_{A_\dr}(\lambda'\vert \lambda, a_\dr) = \delta_{\lambda_\dr, \lambda'}
\end{equation}
where $\lambda_\dr$ is a given hidden variable. Thus, the instrument simply always outputs the same hidden variable, whatever the input, and again, this does not depend on the classical outcome.

\subsubsection{Programmable instruments and OI}

The notion of outcome independence has already appeared in the work of Buscemi et al.~\cite{buscemi2023UnifyingDifferentNotions}, where they introduced a specific type of instrument which has a classical and a quantum input, as well as a classical and a quantum output, altogether called a programmable instruments. The particularity of these instruments is that the production of the classical output is detached from the production of the quantum output, without the possibility to communicate.
It is interesting to notice that some type of compatibility, termed \enquote{classical compatibility}~\cite[Def. 2]{buscemi2023UnifyingDifferentNotions}, is inherently OI.

\section{Non-contextuality for sequential scenarios} \label{sec:non-contextuality}

We now show how the restrictions introduced in Section~\ref{sec:operational-restrictions} relate to non-contextuality, specifically as defined in the sheaf theoretic framework. We start with the definition of non-contextuality, and we show that non-contextual empirical models form a polytope. Not only that, but we also expose the contextual fraction and how it relates to sequential scenarios. First, we define the unique set of instruments in a sequence $S$ as:
\begin{equation} \label{eq:unique_set_of_instruments}
\bar{S} = \left\{A \middle| \exists i \in \mathbb{N}, A_i \in S\right\}
.\end{equation} 
This corresponds to the set of all instruments in a sequence where we removed the indices, and by definition $\bar{S} \in X$ for any $\XMOSEQ$. As an example, if $S = \{A_1, B_2, A_3, C_4, B_5\}$ then we have $\bar{S} = \{A,B,C\}$. We also define a consistent outcome:
\begin{definition}[Consistent outcome] \label{def:consistent-outcome}
	Let a sequential scenario $\XMOSEQ$ and a sequence $S \in \mathcal{S}$. An outcome of that sequence $o_S \in O_S$ is consistent if for every instrument $A_i$ in the sequence $S$ their outcomes are the same independently of their positions in the sequence. 
	More precisely, $o \in O_S$ is consistent if there exists $o_{\bar{S}}\in O_{\bar{S}}$ such that for all $A \in \bar{S}$ and for all $i \in \mathbb{N}$:
	\begin{equation}\label{eq:consistent-outcome}
		    \quad A_i \in S\implies \restr{o}{A_i} = \restr{o_{\bar{S}}}{A}
	.\end{equation}
\end{definition}
For instance, take the sequence $S = \{A_1, B_2, A_3, C_4, B_5\}$, and the outcomes  $o_{S} = (1,0,1,1,0)$ is consistent because $\restr{o_{S}}{A_1} = \restr{o_{S}}{A_3} = 1$ and $\restr{o_{S}}{B_1} = \restr{o_{S}}{B_5} = 0$. On the opposite for the same sequence the outcome $o_{S} = (1,1,0,0,1)$ is not consistent.

We directly show that the outcome $o_{\bar{S}} \in O_{\bar{S}}$ associated to $o_S \in O_S$ is unique:
\begin{lemma}[Uniqueness of $o_{\bar{S}}$]
    For a given sequence $S$, if $o_S \in O_{S}$ is a consistent outcome then there is a unique $o_{\bar{S}} \in O_{\bar{S}}$ associated with it.
\end{lemma}
\begin{proof}
	We prove it by contradiction. Let $o_S \in O_S$ be a consistent outcome for a sequence $S$. Let $o_{\bar{S}}, o'_{\bar{S}} \in O_{\bar{S}}$ such that $o_{\bar{S}} \neq o'_{\bar{S}}$ and they both satisfy Equation~\eqref{eq:consistent-outcome}. Necessarily we have for all $A \in \bar{S}$ and for all $i \in \mathbb{N}$:
	\begin{alignat}{3}
		A_i \in S &\implies \left\{\begin{matrix}
				\restr{o_S}{A_i} &= \restr{o_{\bar{S}}}{A} \\
				\restr{o_S}{A_i} &= \restr{o'_{\bar{S}}}{A}
\end{matrix}\right. \\
			  &\implies \restr{o_{\bar{S}}}{A} = \restr{o'_{\bar{S}}}{A}
	\end{alignat}
    Thus they must be the same, which contradicts $o_{\bar{S}} \neq o'_{\bar{S}}$.
\end{proof}

We now define non-contextuality for sequential scenarios, following closely the definitions in~\cite{abramsky2011SheaftheoreticStructureNonlocality}, but where we allow for repeating instruments within a sequence.
\begin{definition}[Non-contextuality for sequential scenarios] \label{def:non-contextuality-sequential}
    We say that an empirical model $e$ on a sequential scenario $\XMOSEQ$ is non-contextual, if there exist a global probability distribution on outcomes $d\colon O_X \to [0,1]$, such that any marginal in the empirical model can be retrieved as a marginal of the global probability distribution $d$ if the outcomes are consistent. Then for all $S \in \mathcal{S}$ and for all $o_S \in O_S$:
    \begin{equation} \label{eq:non-contextuality-sequential}
        e_S\left( o_S \right)  =  \left\{\begin{matrix*}[l]
 \restr{d}{\bar{S}}(o_{\bar{S}}),\, &\textup{if } o_S \textup{ is consistent}\\
 0,\,  &\textup{otherwise}
\end{matrix*}\right.
    \end{equation}
\end{definition}

\noindent \textit{Remark:} For any sequence $S \in \mathcal{S}$ if there exists another sequence $S' \in \mathcal{S}$ such that $\bar{S} = \bar{S'}$ then they must have the same probability distribution over consistent outcomes. This can be directly inferred from Equation~\eqref{eq:non-contextuality-sequential}.

\subsection{Convexity of non-contextual empirical models} \label{ssec:convexity-nc-e}

It is natural to wonder whether the set of non-contextual behaviour is convex, and whether the convex set can be defined as a polytope or not. We prove that it is in fact a convex set -- which means that a mixture of non-contextual empirical models must yield a non-contextual empirical model. But, it is also a polytope, with extremal points, and we argue that the extremal points of the non-contextual polytope for sequential scenarios are given by the global deterministic assignments on $O_X$.

\begin{theorem} \label{thm:thm-1}
    The set of non-contextual behaviours form a polytope, where the extremal points are given by global deterministic assignments.
\end{theorem}

\begin{proof}
    By definition, we say that en empirical model $e$ on a sequential scenario $\XMOSEQ$ is non-contextual if there exist a global distribution on the outcomes $d\colon O_X \to [0,1]$ and it respects Equation~\ref{eq:non-contextuality-sequential}.
    First we can see that such a global probability distribution can be seen as a mixture of deterministic global assignments $\{ d^{\text{det}_i}\}_i$:
    \begin{equation}
        d = \sum_i \mu_i d^{\text{det}_i}
    \end{equation}
    where $d^{\text{det}_i} \colon O_X \to \set{0,1}$ is a deterministic assignment of outcomes and $\set{\mu_i}_i$ gives a probability distribution over these deterministic assignments.
    Each of these $d^{\text{det}_i}$ give rise to a deterministic empirical model given by, for all $S \in \mathcal{S}$:
    \begin{equation}
        e_S^{\text{det}_i}(o) =  \left\{\begin{matrix*}[l]
 \restr{d^{\text{det}_i}}{\bar{S}}(o_{\bar{S}}),\, &\textup{if } o \textup{ is consistent}\\
 0,\,  &\textup{otherwise}
\end{matrix*}\right.
    \end{equation}
    Then we can rewrite Equation~\eqref{eq:non-contextuality-sequential} as the following:
    \begin{equation}
        e_S(o) =  \left\{\begin{matrix*}[l]
 \sum_i \mu_i \restr{d^{\text{det}_i}}{\bar{S}}(o_{\bar{S}}),\, &\textup{if } o \textup{ is consistent}\\
 0,\,  &\textup{otherwise}
\end{matrix*}\right.
    \end{equation}
    More concisely, this may be written as:
    \begin{equation}
        e_S(o) = \sum_i\mu_i e_S^{\text{det}_i}(o) 
    \end{equation}
    Thus we have just proven that any non-contextual model can be written as a convex mixture of deterministic non-contextual empirical model, which are the vertices of the polytope.
\end{proof}

Because sequential scenarios allow for repeated instruments inside a context, we can wonder whether it impacts the shape of the polytope to add a repeated instrument at the end of a context. For instance, let us consider the following extended sequential KCBS scenario:
\begin{example}[Extended sequential KCBS] \label{ex:extended-KCBS}
    Let the sequential KCBS scenario as per Example~\ref{ex:seq-KCBS}, where we take a different sequence $\mathcal{S}$:
    \begin{equation}
        \begin{split}
            \mathcal{S}_{\text{E-}\KCBS} = \big \{ &\{A_1, B_2, A_3\}, \{B_1, C_2\}, \\
            &\{C_1, D_2\}, \{D_1, E_2\},\\
            &\{E_1, A_2\}\big \}
        \end{split}
    \end{equation}
    where we have simply added $A_3$ at the end of the first sequence.
\end{example}

By construction, the spaces of empirical models, or correlations, of Example~\ref{ex:seq-KCBS} and Example~\ref{ex:extended-KCBS} are different (since more outcomes are involved in Example~\ref{ex:extended-KCBS}). However, since the set of instruments is the same, by Theorem \ref{thm:thm-1} the non-contextual polytopes have the same number of extremal points, that are given by the fixed assignments for each instrument. 

\subsection{Contextual fraction for sequential scenarios}

Once we have an empirical model $e$ on a given sequential scenario $\XMOSEQ$, we can require not only to have a qualitative measure of contextuality but rather a quantitative measure of contextuality, i.e., how contextual is a given empirical model.
In the standard sheaf theoretic contextuality scenario, contextuality can be quantified with the contextual fraction~\cite{abramsky2011SheaftheoreticStructureNonlocality}, and this can be implemented as a linear program~\cite{abramsky2017ContextualFractionMeasure}. 
Thus, we naturally extend the definition of the contextual fraction for sequential scenarios in this section.

In order to define the contextual fraction, we first start by decomposing any sequential empirical model as a convex sum of empirical models:
\begin{equation}
    e = \lambda e^{\text{NC}} + (1-\lambda)e'
\end{equation}
where $e^{\text{NC}}$ is a non-contextual model as per Definition~\ref{def:non-contextuality-sequential} and $e'$ is any empirical model. If one maximizes the weight on $\lambda$, which takes the optimal value $\lambda^*$, then we can define the non-contextual fraction of $e$ as $\NCF(e) \coloneq \lambda^*$, and by definition the contextual fraction is defined as $\CF(e) \coloneq 1 - \NCF(e)$. The contextual fraction was originally designed for empirical models on measurement scenarios (See Section~\ref{sec:meas-scen-to-seq}), but we see from the definition that it naturally extends to empirical models on sequential scenarios. It is a convex measure~\cite{abramsky2017ContextualFractionMeasure} and it is continuous~\cite{vallee2024CorrectedBellNoncontextuality}.

In section~\ref{ssec:convexity-nc-e} we questioned whether additions of repeated instruments changes the NC polytope. In fact, adding instruments, changes the whole scenario, and the question is whether the new scenario is different from the non-repeated scenario. A good question, is whether we have a new difference between quantum and non-contextual. To be more specific, an instance of this problem is the sequential KCBS given in example~\ref{ex:seq-KCBS}, then in using standard quantum theory, can we achieve a higher contextual fraction simply by repeating some instruments ? We phrase this problem as an open question:

\begin{openquestion}
	In the general case, it may be the case that repeating instruments inside a sequential scenario give more or less contextuality. Answering this question would lead to a better understanding of sequential scenarios, and it is left open for future work.
\end{openquestion}

\section{Restricted HVMs and non-contextuality} \label{sec:restriced_HVMs_and_non-contextuality}

In the previous section we have introduced the definition of contextuality, we now show how to derive this notion of contextuality from restricted HVMs.

First we define what is a non-disturbing hidden variable model (ND HVM):
\begin{definition}[ND HVM] \label{def:ND-HVM}
    A HVM $\HVM$ on a sequential scenario $\XMOSEQ$ is non-disturbing, if for any sequence $S \in \mathcal{S}$ and any two instruments in that sequence $A_i \neq B_j \in S$ we have:
    \begin{equation}
        j > i \implies A_i \nd B_j\,.
    \end{equation}
    Where the notation $A_i \nd B_j$ means that $A_i$ does not disturb $B_j$. We will refer to any behaviour of such a hidden variable model as an ND HVM model.
\end{definition}
In plain English, it means that for all sequences within a ND HVM, any instrument in a sequence is such that it will not disturb any other instrument later in that sequence.

Similarly, we refer to an outcome deterministic HVM $\HVM$ on a sequential scenario $\XMOSEQ$, whenever every instrument is outcome deterministic in the sense of Definition~\ref{def:outcome-determinism}:
\begin{definition}[OD HVM]
        A HVM $\HVM$ on a sequential scenario $\XMOSEQ$ is outcome deterministic if for any sequence $S \in \mathcal{S}$, all instruments in such a sequence $A_i \in S$ are OD as in Definition~\ref{def:outcome-determinism}. 
\end{definition}

Finally, we talk of an outcome independent HVM:
\begin{definition}[OI HVM]
    A HVM $\HVM$ on a sequential scenario $\XMOSEQ$ is outcome independent if for any sequence $S \in \mathcal{S}$, all instruments in such a sequence $A_i \in S$ are OI as in Definition~\ref{def:outcome-independence}.
\end{definition}

\subsection{No-disturbance and the compatibility of marginals} \label{ssec:no-disturbance-and-com}

In the sheaf theoretic framework for contextuality, if there exists a HVM which respects both the compatibility of marginals (generalization of no-signalling) and outcome determinism then it is  equivalent to non-contextuality. 
In this paper we have been looking at no-disturbance and NBTS, together they imply the compatibility of marginals. More formally:
\begin{theorem}[ND implies the compatibility of marginals] \label{thm:ND-equiv-COM}
	Let a ND HVM behaviour $h$ on a sequential scenario $\XMOSEQ$, if $h$ respects the NBTS condition then it also respects the compatibility of marginals. 
\end{theorem}
For the full proof see Appendix~\ref{proof:ND-equiv-COM}, the general idea is to use ND as one way no signalling, and causality as the other way no-signalling. With these assumptions, one ends up by respecting the compatibility of marginals.

Theorem~\ref{thm:ND-equiv-COM} has equivalently been phrased in~\cite{clemente2016NoFineTheorem}, with the notion of no-signalling in time together with the arrow of time. Nevertheless, we proved it again to show it in the formalism introduced in this paper. Finally, we want to emphasize that the compatibility of marginals and no-disturbance are not on an equal footing, no-disturbance is a weaker notion, and furthermore no-disturbance refers both to the mathematical property and the physical property, while the compatibility of marginals is a property that arise from various physical properties such as no-signalling or no-disturbance combined with NBTS as shown by Theorem~\ref{thm:ND-equiv-COM}.

\subsection{No-disturbance with outcome determinism is equivalent to non-contextuality for sequential scenarios}

After Theorem~\ref{thm:ND-equiv-COM} shows us that any ND HVM on a sequential scenario implies the compatibility of marginals we now turn to OD ND HVM, and we show that these hidden variables are non-contextual.
\begin{theorem}[ND and OD HVM are equivalent to non-contextuality] \label{thm:ND-and-OD-are-NC}
	Let an empirical model $e$ on a sequential scenario $\XMOSEQ$. This empirical model is non-contextual if and only if there exists a ND OD HVM $\HVM$ respecting the NBTS condition which realizes it. 
\end{theorem}
\noindent The proof is developed in appendix~\ref{proof:ND-and-OD-are-NC}.

We now turn to two seminal examples in the literature for contextuality. One is the KCBS scenario~\cite{klyachko2008SimpleTestHidden} which we already partially covered in Example~\ref{ex:seq-KCBS} and the other is the Peres-Mermin square~\cite{peres1990IncompatibleResultsQuantum,mermin1993HiddenVariablesTwo}. We show below that both examples fit in our definitions and still witness contextuality in sequential scenarios. In both examples below we will make use of Lüders instruments~\cite{busch2009LudersRule}: given a quantum POVM $E^{x} = \{E^{x,y}\}_{y \in O_{E^{x}}}$ its Lüders instrument is defined for all $y \in O_{E^{x}}$ as $\mathcal{I}_{E^{x}}^{y}(\rho) = \sqrt{E^{x,y}} \rho \sqrt{E^{x,y}}$. 
The derivation of the inequalities below can be done following the same reasoning as in standard contextuality scenarios, and we formalize this general connection in section~\ref{sec:meas-scen-to-seq} with Theorem~\ref{thm:induced_sequential_scenarios}. These examples now allow us to interpret these inequalities in the sequential scenario versions.

\begin{example}[KCBS non-contextuality in sequential scenarios] \label{ex:KCBS-NC}
		We start from Example~\ref{ex:seq-KCBS} with the sequential scenario $\XMOSEQ[\textup{KCBS}]$. Suppose that you can operationally justify determinism and non disturbance, i.e. $A \nd B$, $B \nd C$,\ldots, $E \nd A$. Then the following inequality~\cite{klyachko2008SimpleTestHidden,cabello2014GraphTheoreticApproachQuantum} holds by virtue of Theorem~\ref{thm:ND-and-OD-are-NC}:
		\begin{equation} \label{eq:KCBS-ineq}
			\begin{split}
				p(a = b  \vert A_1, B_2) +& p(b = c  \vert B_1, C_2) + p(c = d  \vert C_1, D_2)\\
				&\quad + p(d = e  \vert D_1, E_2) + p(e = a  \vert E_1, A_2) \ge 1
			.\end{split}
		\end{equation} 
		Let us turn to the quantum realization of the KCBS inequality, which as already been derived~\cite{klyachko2008SimpleTestHidden,cabello2013SimpleHardyLikeProof}, and where we shall only show that it applies for instruments as well.
		To each instrument $A$ in $X_{\KCBS}$ we will associate a dichotomic quantum instrument $\mathcal{I}^{A} = \{\mathcal{I}^{A}_{j}\}_{j=0,1}$ where $\mathcal{I}^{A}_{j}(\rho) = P^{A}_{j} \rho P^{A}_{j}$ and each projector $P^{A}_{0}$ is given by $\ketbra{\psi_{A}}{\psi_{A}}$ with:
		\begin{align*}
			\ket{\psi_{A}} &= \frac{1}{N} \begin{pmatrix} 1\\ 0 \\ \sqrt{\cos\theta}  \end{pmatrix} & \ket{\psi_B} &= \frac{1}{N} \begin{pmatrix} \cos\left( 4 \theta \right) \\ \sin\left( 4 \theta \right)  \\ \sqrt{\cos\theta}  \end{pmatrix}\\
			\ket{\psi_{C}} &= \frac{1}{N} \begin{pmatrix} \cos\left( 2\theta \right) \\ -\sin\left( 2\theta \right)  \\ \sqrt{\cos\theta} \end{pmatrix} & \ket{\psi_D} &= \frac{1}{N} \begin{pmatrix} \cos\left( 2\theta \right) \\ \sin\left( 2 \theta \right)  \\ \sqrt{\cos\theta}  \end{pmatrix}\\
			\ket{\psi_E} &= \frac{1}{N} \begin{pmatrix} \cos\left( 4\theta \right) \\ -\sin\left( 4\theta \right)  \\ \sqrt{\cos\theta}  \end{pmatrix} 
		.\end{align*}
	Where $\theta = \sfrac{\pi}{5}$, $N = \sqrt{1 + \cos\theta}$ and:
	\begin{equation}
		P_{1}^{A} =  \mathbbm{1} - P_0^{A} 
	.\end{equation}
	Where $\mathbbm{1}$ is the identity for a qutrit. Since we have pairwise commutation for all projectors then the quantum instruments must be non-disturbing~\cite{heinosaari2010NondisturbingQuantumMeasurements} -- see Definition~\ref{def:quantum_non-disturbing_measurements_from--cite-heinosaari2010nondisturbing} -- $\mathcal{I}^{A} \nd \mathcal{I}^{B}$, \ldots, $\mathcal{I}^{D} \nd \mathcal{I}^{E}$, $\mathcal{I}^{E} \nd \mathcal{I}^{A}$.
	Yet for the state $\rho = \ketbra{\varphi}{\varphi}$ where $\bra{\varphi} = (0, 0, 1)$, we observe a violation of inequality~\eqref{eq:KCBS-ineq}. In fact, in the quantum case, for any two dichotomic instruments $A, B \in X_{\text{KCBS}}$ we have:
	\begin{equation} \label{eq:probability_equality_quantum_instruments}
	p_{Q}(a = b  \mid A_1 B_2) = \Tr\left[ \mathcal{I}^{B}_0 \left( \mathcal{I}^{A}_0\left( \rho \right)  \right) \right]  + \Tr \left[  \mathcal{I}^{B}_1 \left( \mathcal{I}^{A}_1\left( \rho \right)  \right) \right]  
	.\end{equation} 
	Using Equation~\eqref{eq:probability_equality_quantum_instruments} to compute Equation~\eqref{eq:KCBS-ineq} we obtain:
	\begin{equation}
		\begin{split}
			p_{Q}(a = b  \vert A_1, B_2) +& p_{Q}(b = c  \vert B_1, C_2) + p_{Q}(c = d  \vert C_1, D_2)\\
			&\quad + p_{Q}(d = e  \vert D_1, E_2) + p_{Q}(e = a  \vert E_1, A_2) = 5 - 2 \sqrt{5} 
		.\end{split}
	\end{equation} 
	With $5 - 2\sqrt{5} \approx 0.53 < 1$ we obtain a violation of the inequality.
\end{example}

\begin{example}[Peres-Mermin non-contextuality in sequential scenarios] \label{ex:PM-NC-in-sequential-scenarios}
	We show that the Peres-Mermin Square~\cite{peres1990IncompatibleResultsQuantum,mermin1993HiddenVariablesTwo,budroni2022KochenSpeckerContextuality} can be encompassed in this framework if it is done sequentially, following the same steps as the KCBS example above.
	First, we start by ordering the instruments in the contexts, to give sequences, which is also done in Ref~\cite{kirchmair2009StateindependentExperimentalTest} for example. Thus, the sets in $\XMOSEQ[\textup{PM}]$ are:  
	\begin{align}
	    \begin{split}
		X_\textup{PM} &= \left\{ A, B, C, \mathcal{A}, \mathcal{B}, \mathcal{C}, \alpha, \beta, \gamma \right\}
	    \end{split}
	    \\[2ex]
	    \begin{split}
		    \mathcal{S}_\textup{PM} &= \left\{ \left\{A_1, B_2, C_3\right\}, \left\{ \mathcal{A}_1, \mathcal{B}_2, \mathcal{C}_3\right\}, \left\{\alpha_1, \beta_2, \gamma_3 \right\}\right.\\
					    &\qquad \left. \left\{ A_1, \mathcal{A}_{2}, \alpha_3\right\}, \left\{B_1, \mathcal{B}_2, \beta_3 \right\}, \left\{ C_1, \mathcal{C}_2, \gamma_3 \right\}\right\}
	    \end{split}
	    \\[2ex]
	    \begin{split}
		    \forall x \in X\colon \left( O_\textup{PM} \right)_x &= \{-1, 1\}
	    \end{split}
	\end{align}
	Assuming that no-disturbance and outcome determinism holds, from experiments results for instance, then the following inequality holds~\cite{cabello2008ExperimentallyTestableStateIndependent}\footnote{This inequality has been debated in~\cite{krishna2017DerivingRobustNoncontextuality} from a generalized contextuality point of view.}:
	\begin{equation}\label{eq:PM-inequality}
		\begin{split}
			&p( 1  \mid A_1 B_2 C_3) + p( 1  \mid \mathcal{A}_1 \mathcal{B}_2 \mathcal{C}_3) + p( 1 \mid \alpha_1 \beta_2 \gamma_3) \\
			&\quad + p(1  \mid A_1 \mathcal{A}_1 \alpha_1) + p(1  \mid B_1 \mathcal{B}_2 \beta_3) + p(-1  \mid C_1 \mathcal{C}_2 \gamma_3) \le 5
		.\end{split}
	\end{equation} 
	Where $p( \pm 1  \mid S)$ is the probability that product of the outcomes of the instruments in the sequence $S$ to be $\pm 1$.
	The PM square inequality can be violated in quantum theory with observables of the form $O^{A} = P^{A}_0 - P^{A}_{1}$ such that the instruments $\mathcal{I}^{A}$ are defined by Lüders instruments, in the same way as Example~\ref{ex:KCBS-NC}. The observables are the following:
	\begin{align*}
		O^{A} &= \sigma_z \otimes \mathbbm{1} & O^{B} &= \mathbbm{1} \otimes \sigma_z & O^{C} &= \sigma_z \otimes \sigma_z\\
		O^{\mathcal{A}} &= \mathbbm{1} \otimes \sigma_x & O^{\mathcal{B}} &= \sigma_x \otimes \mathbbm{1} & O^{\mathcal{C}} &= \sigma_x \otimes \sigma_x \\
		O^{\alpha} &= \sigma_z \otimes \sigma_x & O^{\beta} &= \sigma_x \otimes \sigma_z & O^{\gamma} &= \sigma_y \otimes \sigma_y 
	.\end{align*}
	Where $\sigma_{x,y,z}$ are the Pauli matrices and $\mathbbm{1}$ is the qubit identity.

	Similarly to the previous example, we see that all measurements are projective and commuting within a sequence $S \in \mathcal{S}$. The associated quantum instruments must therefore be non-disturbing:
	\begin{equation}
		\begin{split}
			\forall S \in \mathcal{S}_{\textup{PM}}\colon \forall A_i, B_j \in S\\
			j > i \implies \mathcal{I}^{A_i} \nd \mathcal{I}^{B_j}
		.\end{split}
	\end{equation} 
	We can now rewrite the probabilities in terms of trace of the instruments $A,B,C \in X$ for any quantum state $\rho$:
	\begin{equation}
		\begin{split}
			p_{Q}( 1  \mid A_1 B_2 C_3 ) &= \Tr\left[ \mathcal{I}^{C}_1 \left( \mathcal{I}^{B}_1 \left( \mathcal{I}^{A}_1 \left( \rho \right)  \right)  \right)  \right] + \Tr\left[ \mathcal{I}^{C}_1 \left( \mathcal{I}^{B}_{-1} \left( \mathcal{I}^{A}_{-1} \left( \rho \right)  \right)  \right)  \right] \\
					     &\quad + \Tr\left[ \mathcal{I}^{C}_{-1} \left( \mathcal{I}^{B}_{-1} \left( \mathcal{I}^{A}_{1} \left( \rho \right)  \right)  \right)  \right] + \Tr\left[ \mathcal{I}^{C}_{-1} \left( \mathcal{I}^{B}_{1} \left( \mathcal{I}^{A}_{-1} \left( \rho \right)  \right)  \right)  \right]
		\end{split}
	.\end{equation}
	Where $p(-1  \mid A_1 B_2 C_3) = 1 - p(1  \mid A_1 B_2 C_3)$.
	Finally, from the quantum example we get the following violation of Inequality~\eqref{eq:PM-inequality}:
	\begin{equation*} \label{eq:PM-inequality-violation}
		\begin{split}
			&p_{Q}( 1  \mid A_1 B_2 C_3) + p_{Q}( 1  \mid \mathcal{A}_1 \mathcal{B}_2 \mathcal{C}_3) + p_{Q}( 1 \mid \alpha_1 \beta_2 \gamma_3) \\
			&\quad + p_{Q}(1  \mid A_1 \mathcal{A}_1 \alpha_1) + p_{Q}(1  \mid B_1 \mathcal{B}_2 \beta_3) + p_{Q}(-1  \mid C_1 \mathcal{C}_2 \gamma_3) = 6
		.\end{split}
	\end{equation*}
	As $6 > 5$ this concludes the example.
\end{example}

\subsection{No-disturbance with outcome independence is equivalent to non-contextuality for sequential scenarios}

In this section we show that in a sequential scenario, an empirical model is non-contextual if there exists a ND OI HVM which realizes it, as long as we restrict ourselves to sequences $S \in \mathcal{S}$ of $2$ instruments or less.
Restricting to sequences of length less that $2$ still covers numerous interesting contextuality scenarios, such as all the n-cycle scenarios~\cite{araujo2013AllNoncontextualityInequalities}. We formalize it with the following theorem:
\begin{theorem}[ND and OI HVM are equivalent to non-contextuality] \label{thm:ND-and-OI-are-NC}
    Let an empirical model $e$ on a sequential scenario $\XMOSEQ$ such that:
    \begin{equation}
        \forall S \in \mathcal{S}\colon \norm*{S} \leq 2
    \end{equation}
    This empirical model is non-contextual if and only if there exists a ND OI HVM $\HVM$ respecting the NBTS condition which realizes it.
\end{theorem}

The full proof is given in the Appendix~\ref{proof:ND-and-OI-are-NC}.

\section{Measurement scenario to sequential scenarios} \label{sec:meas-scen-to-seq}

One of the motivations for studying no-disturbance in sequential scenarios is its operational significance. In this spirit, it is reasonable to ask whether any measurement scenario as described by the sheaf theoretic framework can be transposed to a sequential scenario. We answer by the positive constructively. We moreover show that, for our construction, the non-contextual polytope is conserved, thus the non-contextual inequalities are preserved through what we call the induced sequential scenario. 

\begin{openquestion}
Whether the quantum set is also preserved through this mapping is left open. We strongly believe that it should not be the same, since we are allowing any type of instruments, but at this time we are not able to provide a proof of this claim.
\end{openquestion}

In the following, we first make a quick introduction to measurement scenarios and empirical scenarios in the sheaf theoretic framework for contextuality, then we show how one can map any measurement scenario to a sequential scenario.

\subsection{Measurement scenario à la sheaf theoretic framework}

We restrict ourselves to an informal description of the sheaf theoretic framework, with enough details for the reader to follow. A more thorough description can be found in the original paper~\cite{abramsky2011SheaftheoreticStructureNonlocality} and in subsequent papers~\cite{vallee2024CorrectedBellNoncontextuality,abramsky2017ContextualFractionMeasure}.

A \textit{measurement scenario} is characterized by the tuple $\XMO$ where \begin{enumerate*}[label=(\roman*)]
  \item $X$ is a set of \textit{measurement} labels;
  \item $\mathcal{M} \subseteq \mathcal{P}(X)$ is the set of maximal contexts and forms a cover of X;
  \item $O = (O_x)_{x \in X}$ and $O_x$ is the outcome set corresponding to measurement $x \in X$.
\end{enumerate*}

One can define an \textit{empirical model} $e$ on a measurement scenario $\XMO$ as a family of probability distributions $e = (e_C)_{C\in\mathcal{M}}$, in the same spirit as an empirical model on a sequential scenario.

	The compatibility of marginals, referring to the operational property of no-signalling, parameter independence or \enquote{no-disturbance}, is substantially simplified in measurement scenarios. An empirical model is said to respect the compatibility of marginals if for any two contexts $C_1$ and $C_2$ we have, $\restr{e_{C_1}}{C_1 \cap C_2} = \restr{e_{C_2}}{C_1 \cap C_2}$ where the notation $\restr{e_{C}}{U}$ stands for the marginalization of the probability distribution $e_C$ to $U \subseteq C$. 

\textit{Non-contextuality} is the fact that an empirical model $e$ on a measurement scenario $\XMO$ can be described by a global probability distribution on the outcomes $d\colon O_X \to [0,1]$, i.e., all context wise probability distributions $e_C\colon O_C \to [0,1]$ can be seen as the marginal of $d$, or, mathematically, $e_C = \restr{d}{C}$.

We follow up with the definition of a \textit{hidden variable model}. A HVM on a measurement scenario $\XMO$ is a triple $\tuple{\Lambda, \mu, (h^\lambda)_{\lambda \in \Lambda}}$, where \begin{enumerate*}[label=(\roman*)]
  \item $\Lambda$ is a finite set of hidden variables;
  \item $\mu\colon \Lambda \to [0,1]$ is the distribution over hidden variables;
  \item for all $\lambda$ we have $h^\lambda$ as an empirical model on $\XMO$.
\end{enumerate*}
This gives rise to an overall empirical model given as the average over the various hidden variables $h \coloneq \sum_{\lambda \in \Lambda} \mu(\lambda) h^\lambda$.
Similarly to empirical models, we can impose that a given hidden variable behaviour $h^\lambda$ must be parameter independent, i.e. behaviours which respect the compatibility of marginals, such that for all $C_1, C_2 \in \mathcal{M}$ we have $\restr{h^\lambda_{C_1}}{C_1 \cap C_2} = \restr{h^\lambda_{C_2}}{C_1 \cap C_2}$.
Imposing such a constraint on each of the hidden variable behaviours lead to an empirical model which respect the compatibility of marginals, such that for all $C_1, C_2 \in \mathcal{M}$:
\begin{equation}
\begin{split}
    &\forall \lambda \in \Lambda \colon \restr{h^\lambda_{C_1}}{C_1 \cap C_2} = \restr{h^\lambda_{C_2}}{C_1 \cap C_2} \\
    &\quad \implies \restr{h_{C_1}}{C_1 \cap C_2} = \restr{h_{C_2}}{C_1 \cap C_2}
\end{split}
\end{equation}
Finally, we can define a deterministic, or sharp, hidden variable behaviour, such that $\forall \lambda \colon \forall C \in \mathcal{M}\colon h^\lambda_C(t) = \delta_{t,t'}$ where $t,t' \in O_C$ and $\delta$ is the Dirac distribution.

We write below the conclusion reached in~\cite[Prop. 3.1 and Th. 8.1]{abramsky2011SheaftheoreticStructureNonlocality}:
\begin{proposition}
    An empirical behaviour $e$ on a measurement scenario $\XMO$ is non-contextual if and only if it is realizable by a parameter independent deterministic HVM.
\end{proposition}

% More about the notion of determinism and parameter independence, and their relaxation, can be found in~\cite{vallee2024Corrected}.

\subsection{Mapping measurement scenarios to sequential scenarios} \label{ssec:ms-to-ss}

From the definitions given in the previous section, we are able to show that any measurement scenario can be converted to a sequential scenario, where the non-contextual set is preserved. This means the notion of no-signalling is replaced by no-disturbance and measurements are turned into instruments.

Let an arbitrary measurement scenario $\XMO$, with measurement labels $X$, contexts $\mathcal{M}$ and outcomes $O = (O_x)_{x \in X}$. For simplicity we will assume to be the same for each measurement, i.e. $\forall x \in X\colon O_x = O$.
From such a measurement scenario, we can define a sequential scenario where each measurement is promoted to an instrument.
In essence, this does not change the scenario, because a scenario is only concerned about the label of the measurements, and not their nature, but sequential scenarios inherently need a notion of causal ordering, which we reflect on the contexts by transforming them to sequences.
Thus, let us decide of an arbitrary ordering of the instruments within the contexts, in other words for each context $C$ the instruments inside $A \in C$ acquire a new index $\ell$ such that $A_\ell$ represents the instrument at step $\ell$. This new ordered context forms a sequence, which we write $S$ and which we relate to $C$ by the notation $S \simeq C$. 
Practically, the ordering may be chosen arbitrarily, or decided such that the conditions of no-disturbance are best respected. With these changes, one can define a sequential scenario $\XMOSEQ[s]$, with $X_s = X$, $\mathcal{S}_s = \left\{ S \right\}_{S \simeq C, C \in \mathcal{M}}$ and the classical outcomes remain unchanged $O_s = O$.
Sequential scenarios obtained by the above steps from a measurement scenario are called \textit{induced} sequential scenarios:
\begin{definition}[Induced sequential scenario]
    A sequential scenario $\XMOSEQ[s]$ is called an induced sequential scenario if there exist a measurement scenario $\XMO$ such that
    \begin{enumerate*}[label=(\roman*)]
        \item $X_s = X$;
	\item $\left\lvert \mathcal{M} \right\rvert = \left\lvert \mathcal{S}_s \right\rvert$;
        \item for all $S \in \mathcal{S}_s$, there exists $C \in \mathcal{M}$ such that $S \simeq C$ up to the index of the instruments and $O_C = O_S$.
    \end{enumerate*}
\end{definition}

For a given measurement scenario, there exists multiple possible induced sequential scenario, where the various ordering of the instruments make the difference. However the opposite is not true, there exists only one measurement scenario associated to an induced sequential scenario, which basically boils down to the fact that the contexts must be the same as the sequential scenario up to the ordering labelling.

We now prove that such an induced sequential scenario preserves the notion of non-contextuality: if an empirical model is non-contextual in the measurement scenario then it is non-contextual in the sequential scenario and vice versa. By Theorem~\ref{thm:ND-and-OD-are-NC} this implies that the existence of a deterministic parameter independent hidden variable model in the measurement scenario is equivalent to the existence of a deterministic non-disturbing hidden variable model in the induced sequential scenario.

\subsection{Non-contextuality holds in induced sequential scenarios}

Given the way induced sequential scenarios are built, it is expected that the non-contextual sets coincide, and we formalize it below:
\begin{theorem} \label{thm:induced_sequential_scenarios}
    Let an empirical model $e$ on a measurement scenario $\XMO$ and an induced sequential scenario $\XMOSEQ$ built from it as described in Section~\ref{ssec:ms-to-ss}. Then if this empirical model $e$ is NC in the measurement scenario, then it is NC in the induced sequential scenario and vice versa.
\end{theorem}

\begin{proof}
    We start by the first implication, that NC in measurement scenarios implies NC in induced sequential scenarios.
    Let a NC empirical model $e$ on the measurement scenario $\XMO$. Then there exists a global distribution $d\colon O_X \to [0,1]$ such that for each context, the empirical scenario is retrieved as the marginal of that global probability distribution $e_C = \restr{d}{C}$. 
    
    We define an analogous empirical model $e_s$ on the sequential scenario $\XMOSEQ[s]$ as the same empirical model, i.e. it is the same family of probability distributions $e_s \coloneq e$. The only modification in the contexts is the addition of a label to the measurements $A \in C \mapsto A_\ell \in S$. The following is valid for any labelling, as long as the label $\ell$ appears only once in a given sequence, and $\ell \le \norm*{C}$.
    In order to show that the empirical model $e_s$ satisfies the notion of non-contextuality described in Definition~\ref{def:non-contextuality-sequential}, we first note that every sequence $S$ is the same as the set of unordered measurements: $S = \bar{S}$. This holds by definition of our mapping. Secondly, for any sequence $S$, the outcomes $o \in O_S$ are consistent as in Definition~\ref{def:consistent-outcome}, since we have the property $S = \bar{S}$\footnote{$S = \bar{S}$ implies that $O_{S} = O_{\bar{S}}$ and thus $\forall o \in O_{S}$ there must exist $o' \in O_{\bar{S}}$ which is simply defined as $o' = o$}. Because we have that any $o \in O_{S}$ is consistent, we can rewrite Equation~\eqref{eq:non-contextuality-sequential} as:
    \begin{equation}
        e_{S}(o) = \restr{d_s}{\bar{S}}(o_{\bar{S}})
    .\end{equation}
    Which holds for all $S\in \mathcal{S}$ and all $o \in O_{S}$. One can see that $\bar{S} = C$, such that if we take $d_s \equiv d$ we indeed retrieve Equation~\eqref{eq:non-contextuality-sequential}.

    For the other implication, that a NC empirical model in an induced sequential scenario give rise to a NC empirical model in a measurement scenario, the proof follows the same structure.
    Let a non-contextual empirical model $e_s$ in a sequential scenario $\XMOSEQ$. Then there must exist a global probability distribution on the outcomes $d_s\colon O_X \to [0,1]$ such that $e_{S}(o) = \restr{d}{\bar{S}}(o_{\bar{S}})$. Note here that we use the fact previously shown that all outcomes are consistent and $\forall S \in \mathcal{S}\colon S = \bar{S}$. Thus, setting a global distribution in measurement scenario $d \coloneq d_s$, we find that the empirical model in the measurement scenario $e \coloneq e_s$ is also NC, i.e., for all $S \in \mathcal{S}$ and for all $o\in O_S$:
    \begin{subequations}
	    \begin{align}
		e_{S}(o) &= \restr{d_s}{\bar{S}}(o_{\bar{S}})\\
		&= \restr{d}{C}(o_{C}) \\
		&= e_{C}(o)
	    .\end{align}    
    \end{subequations}
    We thus conclude that the empirical model $e$ on the measurement scenario must be non-contextual if $e_s$ is non-contextual.
\end{proof}

Although not surprising, Theorem~\ref{thm:induced_sequential_scenarios} has important consequences, as it allows any contextual scenario in the standard framework to be translated, along with the associated inequalities, measures of contextuality (the contextual fraction) and tools that come with them. In particular, using \cite{abramsky2017ContextualFractionMeasure} one can take experimental data from any sequential version of a contextual inequality and check if it is non-contextual via linear programming. Through Theorems~\ref{thm:ND-and-OD-are-NC} and~\ref{thm:ND-and-OI-are-NC} we have operational interpretations for the possible sequential HVMs.

\section{Discussion and conclusion}

\subsection{Discussion}

We now turn to the discussion of two references which have common objective with this paper, and show how they are related.

\paragraph*{Relation to Ref~\cite{guhne2010CompatibilityNoncontextualitySequential}:} The authors of Reference~\cite{guhne2010CompatibilityNoncontextualitySequential,kirchmair2009StateindependentExperimentalTest} were looking at sequential measurements and the relation to non-contextuality with the assumption of both OD and non-invasiveness. Since this is related to our current work, we investigate and clarify the relation to this paper.
Non-invasiveness means that the state is not disturbed by compatible measurements, which is a particular case of no-disturbance, where the transfer function is simply a Dirac function $\Gamma(\lambda'\vert \lambda, k) = \delta_{\lambda',\lambda}$. In~\cite{guhne2010CompatibilityNoncontextualitySequential}, the authors propose multiple approaches to relax some notion of compatibility. We propose in the following to use only the weaker assumption of no-disturbance instead of the notion of compatibility they introduced for simplifying the notation. They relax the assumption of non-invasiveness by introducing the probability to flip, which they write as:
\begin{equation*} \label{eq:pflip_original_notation}
	\begin{split}
		p^{\text{flip}}[AB] &= p\left[ (b_1 = 1 \mid B_1) \text{ and } (b_2 = 0  \mid A_1 B_2) \right] \\
				      &\quad+ p\left[ (b_1 = 0 \mid B_1) \text{ and } (b_2 = 1  \mid A_1 B_2) \right]
	.\end{split}
\end{equation*} 
Where the instrument $B$ has outcomes $O_B = \left\{0, 1\right\} $. We have interpreted the \enquote{and} as \enquote{at the same time}, which is why they argue that it is not a measurable quantity.
Adapted to our notation this gives:
\begin{equation} \label{eq:p-flip}
	    p^{\textup{flip}}[AB] = \sum_{\lambda_0, \lambda_1, k, k'}\mu(\lambda_0) \xi_{B}(k\mid\lambda_0)\xi_{A}(k'\mid \lambda_0) \Gamma^\varepsilon_A(\lambda_1 \mid \lambda_0, k')\xi_{B}(1-k \mid \lambda_1)
\end{equation}
where $A \nd B$. 

In their papers, the authors do not relax OD, but only focus on the relaxation of compatibility -- here associated to ND. If ND holds perfectly, then $p^{\textup{flip}}[AB]$ is necessarily $0$. The relaxation of ND necessarily comes from the transfer function $\Gamma$, thus one can introduce a noisy version:
\begin{equation}
        \Gamma^\varepsilon_A(\lambda_1 \mid \lambda_0, k') = (1-\varepsilon)\Gamma_A(\lambda_1 \mid \lambda_0, k') + \varepsilon \Gamma_A^{\textup{noisy}}(\lambda_1 \mid \lambda_0, k')
\end{equation}
where $\Gamma_A^{\textup{noisy}}(\lambda_1 \mid \lambda_0, k')$ is a function that chooses a $\lambda_1$ such that any non-disturbing measurement outcome is flipped. The existence of such function is not assured, however it follows the spirit of~\cite{guhne2010CompatibilityNoncontextualitySequential}, and the notion of flipping. With this we have $p^{\textup{flip}}[AB] = \varepsilon$.

However, as argued in~\cite{guhne2010CompatibilityNoncontextualitySequential}, this might not be directly experimentally testable since both $\xi_{B}(k\mid\lambda_0)$ and $\xi_{A}(k'\mid \lambda_0)$ depend on $\lambda_0$ in Equation~\eqref{eq:p-flip} which is not possible because one has to chose which measurement comes first. We can still follow through the paper with our approach, by introducing a new quantity $p^{\textup{err}}[BAB]$:
\begin{equation}
        p^{\textup{err}}[BAB] = \sum_{\substack{\lambda_0, \lambda_1, \lambda_2 \\ k ,k' \in \mathbb{B}}}\mu(\lambda_0) \xi_{B}(k\mid \lambda_0) \Gamma_{B}(\lambda_1 \mid \lambda_0, k) \xi_{A}(k' \mid \lambda_1)\Gamma_{A} (\lambda_2 \mid \lambda_1, k') \xi_{B}(1-k \mid \lambda_2)
\end{equation}
where again $p^{\textup{err}}[BAB] = 0$ as long as OD and ND are respected. The following inequality holds:
\begin{equation}
    p^{\textup{flip}}[AB] \leq p^{\textup{err}}[BAB]\,,
\end{equation}
as we have:
\begin{equation}
    p^{\textup{err}}[BAB] - p^{\textup{flip}}[AB] = (1 - 2\varepsilon)\varepsilon
\end{equation}
and by definition $\varepsilon \leq 0.5$, otherwise it just amounts to flipping the outcomes of the measurements and take $\varepsilon' = 1 - \varepsilon$.
As a conclusion, we acknowledge that there are more approaches used in~\cite{guhne2010CompatibilityNoncontextualitySequential} to account for errors. However, we have decided to restrict ourselves to the first ones to give an example of how one can retrieve existing results with the framework we have developed in this paper.

\paragraph*{Relation to~\cite{abramsky2023CombiningContextualityCausality,searle2024CorrelationCombinatoricsCausal}}

In a recent work Abramsky et al.~\cite{abramsky2023CombiningContextualityCausality,searle2024CorrelationCombinatoricsCausal} developed a framework for causal scenarios, with sequential measurements closely related to this work.
The major differences are twofold, first we emphasize on hidden variable models, with state transformations and operational assumptions.
Second, we allow for repeated instruments, i.e. a sequence can have the same instrument multiple time, at the cost of getting away of the sheaf theoretic mathematical framework.
Furthermore, in this paper, we investigate in more depth the notion of no-disturbance, and which constraints should one put on the HVMs in order to witness non-contextuality. On the other hand, we have not developed the sheaves and categorical approach at the moment which is why we believe that both works are likely complimentary, and it is left as an open avenue whether they can be unified in a more general framework.

While our statement is centred around operational notions and HVMs, a different approach has been taken in Searle's thesis~\cite{searle2024CorrelationCombinatoricsCausal}, where instead they consider a finite memory of the past measurement settings and eventually outcomes, all of it within a sheaf theoretic approach. A central notion is in fact $\mathbb{K}_k$-loopback, which refers to the $k$ previous measurement settings being held in memory, a notion which we have not tackled in this work. Withal, we believe without formal proof, that our definition of a non-contextual hidden variable model for sequential scenario is equivalent to $k=0$ memory. Future interesting avenues would be to look at other operational notions that could motivate different $k$ values of  $\mathbb{K}_k$-lookback, thus following up on the works on the memory cost of some correlations~\cite{hoffmann2018StructureTemporalCorrelations,kleinmann2011MemoryCostQuantum,budroni2019MemoryCostTemporal,budroni2019ContextualityMemoryCost,searle2024CorrelationCombinatoricsCausal}. We should add that Theorem~\ref{thm:induced_sequential_scenarios} is likely a special case of~\cite[Theorem 5.3.1]{searle2024CorrelationCombinatoricsCausal}, if we assume that in fact it corresponds to choosing the right value  $k$ for the memory.

\subsection{Conclusion}

To summarize, we have investigated sequential scenarios and their relation to non-contextuality. We first defined sequential scenarios and empirical behaviours, and subsequently we have looked at the definition of a HVM for sequential scenarios. We proposed some restrictions on this HVM that led to an equivalence with the notion of non-contextuality. Furthermore, we also proved that the non-contextual set is a polytope for any sequential scenarios, with a well-defined notion of contextual fraction.
Finally, we propose a way to map any measurement scenario to a sequential scenario.

Together this represents a powerful framework for considering non-classical behaviour in sequential scenarios. On the one hand we have a clear formal definition of the non-contextual polytope, which allows many tools to be applied, such as linear programming to certify the violation of an inequality given some experimental data \cite{abramsky2017ContextualFractionMeasure}. On the other hand, in Theorems~\ref{thm:ND-and-OD-are-NC} and~\ref{thm:ND-and-OI-are-NC} we have clear interpretation that violation would imply the statistics cannot be explained by HVMs that are both non-disturbing and outcome deterministic for Theorem~\ref{thm:ND-and-OD-are-NC}, or HVMs that are both non-disturbing and outcome independent for Theorem~\ref{thm:ND-and-OI-are-NC}. 

We leave open many avenues for future research in this work, amongst which, the relaxation determinism and no-disturbance, in the same spirit as~\cite{vallee2024CorrectedBellNoncontextuality}. We believe it is an interesting path since, as clearly pointed out in~\cite{guhne2010CompatibilityNoncontextualitySequential}, the need for noise robust inequalities in sequential scenarios is required to pursue realistic experiments. The hidden variable theory proposed in this work is a good starting point for such relaxations.

Another line of research follows along the line of operational contextuality~\cite{spekkens2005ContextualityPreparationsTransformations}, where one could see sequential instruments as a causal network~\cite{chaves2021CausalNetworksFreedom}, and see how that relates to operational contextuality.

Finally, it is not clear to the authors, how the mapping proposed in Section~\ref{sec:meas-scen-to-seq} works for the quantum case, i.e., is any quantum behaviour in the sequential case mapped to a quantum behaviour in the Bell scenario? This goes along the lines of compiled XOR games~\cite{catani2024ConnectingXORXOR,baroni2024QuantumBoundsCompiled}, and it is partially answered in~\cite{baroni2025QuantitativeQuantumSoundness}. We believe that the reverse mapping does not exist since there is a variety of instruments that are acting non-trivially on the state, however this is still to be proven, and finding what exact subset it matches is of interest on its own.

\section{Acknowledgements}

We acknowledge the help of Matilde Baroni, Davide Rolino and Marco Túlio Quintino who indirectly improved many parts of this paper through formal and informal discussions. We would also like to thank Gaël Massé who took the time to proofread the paper and helped enhance the writing style.
We acknowledge financial support from the French national quantum initiative managed by Agence Nationale de la Recherche in the framework of France 2030 through project EPIQ, with the reference ANR-22-PETQ-0007.

\clearpage

\bibliographystyle{plainnat}
\bibliography{biblio}

\clearpage

\appendix

\section{Proofs}

\subsection{Proof of Theorem~\ref{thm:ND-equiv-COM}} \label{proof:ND-equiv-COM}

\begin{proof}
    Let an arbitrary hidden empirical model $h$ from a ND HVM $\HVM$ on a sequential scenario $\XMOSEQ$, any two sequences $S, S' \in \mathcal{S}$, and let two index sets $I_S$ and $I_{S'}$ such that:
    \begin{align}
        S &= \left \{ A_k  \middle| k \le I_S, A \in X  \right \} \\
        S' &= \left \{ A_\ell  \middle| \ell \le I_{S'}, A \in X \right \} 
    .\end{align}
    We define $N = \left\lvert S \right\rvert$ and $N' = \left\lvert S' \right\rvert$.
    Let us take the case where there exists an instrument $A$ for which there exists indices $\tilde{k} \in I_S$ and $\tilde{\ell} \in I_{S'}$ such that:
    \begin{align}
	    A_{\tilde{k}} &\in S\\
	    A_{\tilde{l}} &\in S'
    .\end{align}

    We now check that the compatibility of marginal is respected for any hidden empirical model $h^{\lambda_0} = (h^{\lambda_0}_S)_{S \in \mathcal{S}}$\footnote{Here we use the notation $\lambda_0$ instead of  $\lambda$ as it refers to the hidden variable that is given in the sequence, but will change in its course.}:
	\begin{subequations}
		\begin{align}
			\label{eq:NBTS-decomp} \restr{h_S^{\lambda_{0}}}{A_{\tilde{k}}}(a_{\tilde{k}}) &= \sum_{\substack{\lambda_1, \dots, \lambda_{N+1} \\ a_k \neq a_{\tilde{k}}}} \prod_{k=1}^{N+1} \xi_{A_k}(a_k \mid \lambda_{k-1}) \Gamma_{A_k}(\lambda_k \mid \lambda_{k-1}, a_k) \\[2ex]
			\begin{split} \label{eq:split-proof-ND-equiv-COM}
			    &= \sum_{\substack{\lambda_1, \dots, \lambda_{N+1} \\ a_k \neq a_{\tilde{k}}}} \left [\prod_{k=1}^{\tilde{k}-1} \xi_{A_k}(a_k \mid \lambda_{k-1}) \Gamma_{A_k}(\lambda_k \mid \lambda_{k-1}, a_k) \right ] \\
			    &\qquad\qquad\qquad \times \xi_{A_{\tilde{k}}}(a_{\tilde{k}} \mid \lambda_{\tilde{k}-1}) \Gamma_{A_{\tilde{k}}}(\lambda_{\tilde{k}} \mid \lambda_{\tilde{k}-1}, a_{\tilde{k}}) \\
			    &\qquad\qquad\qquad \times \left [\prod_{k=\tilde{k}+1}^{N+1} \xi_{A_k}(a_k \mid \lambda_{k-1}) \Gamma_{A_k}(\lambda_k \mid \lambda_{k-1}, a_k) \right ]
			\end{split} \\[2ex]
			\begin{split} \label{eq:elimination-rhs-proof-ND-equiv-COM}
			    &= \sum_{\substack{\lambda_1, \dots, \lambda_{\tilde{k} - 1} \\ a_k \neq a_{\tilde{k}}}} \left [\prod_{k=1}^{\tilde{k}-1} \xi_{A_k}(a_k \mid \lambda_{k-1}) \Gamma_{A_k}(\lambda_k \mid \lambda_{k-1}, a_k) \right ] \\
			    &\qquad\qquad\qquad \times \xi_{A_{\tilde{k}}}(a_{\tilde{k}} \mid \lambda_{\tilde{k}-1})
			\end{split} \\[2ex]
			\begin{split} \label{eq:ND-p1-proof-ND-equiv-COM}
			    &= \sum_{\substack{\lambda_1, \dots, \lambda_{\tilde{k} - 1} \\ a_k \neq a_{\tilde{k}}}} \left [\prod_{k=1}^{\tilde{k}-2} \xi_{A_k}(a_k \mid \lambda_{k-1}) \Gamma_{A_k}(\lambda_k \mid \lambda_{k-1}, a_k) \right ] \\
			    &\qquad\qquad\qquad \times \xi_{A_{\tilde{k} - 1}}(a_{\tilde{k} - 1} \mid \lambda_{\tilde{k} - 2}) \Gamma_{A_{\tilde{k} - 1}}(\lambda_{\tilde{k} - 1} \mid \lambda_{\tilde{k} - 2}, a_{\tilde{k} - 1}) \xi_{A_{\tilde{k}}}(a_{\tilde{k}} \mid \lambda_{\tilde{k}-1})
			\end{split} \\[2ex]
			\begin{split} \label{eq:ND-p2-proof-ND-equiv-COM}
			    &= \sum_{\substack{\lambda_1, \dots, \lambda_{\tilde{k} - 2} \\ a_k \neq a_{\tilde{k}}}} \left [\prod_{k=1}^{\tilde{k}-2} \xi_{A_k}(a_k \mid \lambda_{k-1}) \Gamma_{A_k}(\lambda_k \mid \lambda_{k-1}, a_k) \right ] \\
			    &\qquad\qquad\qquad \times \xi_{A_{\tilde{k}}}(a_{\tilde{k}} \mid \lambda_{\tilde{k}-2})
			\end{split} \\[2ex]
			&= \xi_{A_{\tilde{k}}}(a_{\tilde{k}} \mid \lambda_0)
		.\end{align}	
	\end{subequations}
	Where $A_k$ is the $k$\textsuperscript{th} measurement of context $S$. We start from the NBTS assumption in Equation~\eqref{eq:NBTS-decomp}. Then we split the instruments that come before and after $A_{\tilde{k}}$ in Equation~\eqref{eq:split-proof-ND-equiv-COM}, and because of the summation we can directly eliminate the second product in Equation~\eqref{eq:elimination-rhs-proof-ND-equiv-COM}, finally we just iteratively apply ND in Equation~\eqref{eq:ND-p1-proof-ND-equiv-COM} and~\eqref{eq:ND-p2-proof-ND-equiv-COM}. We can apply iteratively the ND relation because by definition of a ND HVM we have:
    \begin{equation}
        \forall j < m\colon A_j \nd A_m
    \end{equation}

    We can also apply this same procedure for the sequence $S'$, and it is quite clear that we obtain the same result. Hence:
    \begin{subequations}
	    \begin{align}
		\restr{h_S^{\lambda_{0}}}{A_{\tilde{k}}}(a_{\tilde{k}}) &= \xi_{A_{\tilde{k}}}(a_{\tilde{k}} \mid \lambda_0) \\
		&= \xi_{A_{\tilde{\ell}}}(a_{\tilde{\ell}} \mid \lambda_0) \\
		&= \restr{h_{S'}^{\lambda_{0}}}{A_{\tilde{\ell}}}(a_{\tilde{\ell}})
	    \end{align}
    \end{subequations}
    Concluding the proof on the compatibility of marginals for NBTS and ND HVMs.
\end{proof}

\subsection{Proof of Theorem~\ref{thm:ND-and-OD-are-NC}} \label{proof:ND-and-OD-are-NC}

In order to prove Theorem~\ref{thm:ND-and-OD-are-NC}, we state the following useful lemma:
\begin{lemma} \label{lemma:ND-for-OD}
    Let two deterministic instruments $A$ and $B$ performed in sequence such that $A \nd B$. Then the following holds:
    \begin{equation}
        \begin{split}
            &\forall \lambda \in \Lambda\colon \forall \lambda' \in \supp\big \{\Gamma_{A}(\cdot \mid \lambda, a)\big \} \colon \forall b \in O_{B} \colon\\
            &\quad \xi_{B}(b\mid \lambda') = \xi_{B}(b\mid \lambda) 
        \end{split}
    \end{equation}
    where $\supp\big \{\Gamma_{A}(\cdot \mid \lambda, a)\big \}$ stands for the support of the transfer function $\Gamma_{A}$, i.e., when $\Gamma_{A}$ is non-zero.
\end{lemma}
\begin{proof}
    Let us denote by $a^{\lambda_0} \in O_A$ the outcome such that $\xi_{A}(a^{\lambda_0} \mid \lambda_0) = 1$. The existence and uniqueness of such an outcome is given by the OD assumption. By using Equation~\eqref{eq:no-disturbance} on no-disturbance we obtain:
    \begin{equation}
        \forall \lambda_0\colon \forall b \colon \sum_{\lambda_1} \Gamma_{A}(\lambda_1 \mid \lambda_0, a^{\lambda_0}) \xi_{B}(b \mid \lambda_1) = \xi_{B}(b \mid \lambda_0)
    \end{equation}
    We now proceed by enumeration, since $\xi_{B}(b \mid \lambda_0) \in \{0,1\}$.
    First we assume $\xi_{B}(b \mid \lambda_0) = 0$, then necessarily
    \begin{equation}
        \forall \lambda_1 \in \supp(\Gamma_{A}) \colon \xi_{B}(b \mid \lambda_1) = 0
    \end{equation}
    otherwise we would have 
    \begin{equation}
		\begin{split}
			\Gamma_{A}(\lambda_1 \mid \lambda_0, a^{\lambda_0}) \xi_{B}(b \mid \lambda_1) &> 0 \\
												      &\neq \xi_{B}(b \mid \lambda_0)
		.\end{split}
    \end{equation}
    We can extend similarly for $\xi_{B}(b \mid \lambda_0) = 1$:
    \begin{equation}
        \forall \lambda_1 \in \supp(\Gamma_{A}) \colon \xi_{B}(b \mid \lambda_1) = 1
    \end{equation}
    Otherwise we have
    \begin{align}
        \sum_{\lambda_1} \Gamma_{A}(\lambda_1 \mid \lambda_0, a^{\lambda_0}) \xi_{B}(b \mid \lambda_1) &< 1 \\
        &\neq \xi_{B}(b \mid \lambda_0)
    \end{align}
\end{proof}

\noindent We now turn to the proof of Theorem~\ref{thm:ND-and-OD-are-NC}:
\begin{proof}
    Let an empirical model $e$ on a sequential scenario $\XMOSEQ$ and let any context $S \in \mathcal{S}$ with $N = \norm*{S}$. In full generality, we denote the instruments in the context $S$ by $A_i$ where $i$ is the order of the instrument, and each instrument is in $X$.
    Suppose now that we have an OD ND HVM $\HVM$. Then we write $o \in O_S$ as $o=(a_1, a_2, \ldots, a_N)$, and we have:
    \begin{equation}
        e_S(o) = \sum_{\lambda_0,\dots,\lambda_{N+1}} \mu(\lambda_0) \prod_{j=1}^{N}\xi_{A_j}(a_j \mid \lambda_{j-1}) \Gamma_{A_j}(\lambda_j \mid \lambda_{j-1}, a_j)
    \end{equation}
    where $\forall j=1,\dots,N$ we have $a_j$ to be the outcome of instrument $A_j$. With some rearranging, we obtain:
    \begin{equation}
        e_S(o) = \sum_{\lambda_0,\dots,\lambda_{N}} \mu(\lambda_0) \xi_{A_1}(a_1 \mid \lambda_0) \prod_{j=2}^{N} \Gamma_{A_{j-1}}(\lambda_{j-1} \mid \lambda_{j-2}, a_{j-1}) \xi_{A_j}(a_j \mid \lambda_{j-1})
    \end{equation}
    We now use lemma~\ref{lemma:ND-for-OD}, because we have for all $j$ that $A_{j-1} \nd A_{j}$ we can rewrite the above equation as:
    \begin{equation}
        \begin{split}
            e_S(o) =& \sum_{\lambda_0,\dots,\lambda_{N}} \mu(\lambda_0) \xi_{A_1}(a_1 \mid \lambda_0) \xi_{A_2}(a_2 \mid \lambda_0) \\
            &\quad \times \left [\prod_{j=2}^{N-1} \Gamma_{A_{j-1}}(\lambda_{j-1} \mid \lambda_{j-2}, a_{j-1}) \xi_{A_{j+1}}(a_{j+1} \mid \lambda_j) \right ] \\
            &\quad {}\times \Gamma_{A_{N-1}}(\lambda_{N-1} \mid \lambda_{N-2}, a_{N-1})
        \end{split}
    \end{equation}
    By acting iteratively we end up with:
    \begin{equation}
        e_S(o) = \sum_{\lambda_0,\dots,\lambda_{N}} \mu(\lambda_0)  \left [ \prod_{j=1}^{N} \xi_{A_j}(a_j \mid \lambda_0) \right ] \left [\prod_{j=1}^{N-1} \Gamma_{A_j}(\lambda_j \mid \lambda_{j-1}, a_j) \right ]
    \end{equation}
    Where the transfer functions $\Gamma_{A_j}$ can be eliminated by the sum over the right $\lambda_j$:
    \begin{equation} \label{eq:eso_proof_thm_complete}
        e_S(o) = \sum_{\lambda_0} \mu(\lambda_0)  \left [ \prod_{j=1}^{N} \xi_{A_j}(a_j \mid \lambda_0) \right ]
    \end{equation}
    We now show that there exist a global probability distribution on outcomes $d\colon O_X\to [0,1]$ which respects the definition of non-contextuality given in Definition~\ref{def:non-contextuality-sequential}.
    First, if $o$ is not consistent, it means that there exists an instrument $A$ which appears twice in a sequence, at step $i$ and step $j$, but which have different outcomes. This implies $\xi_{A_i}(a_i \mid \lambda_0) \xi_{A_j}(a_j \mid \lambda_0) = 0$ since we have deterministic response functions one of them must be $0$. In terms of the value of $e_S(o)$, we deduce from Equation~\eqref{eq:eso_proof_thm_complete}:
    \begin{equation}
        \begin{split}
            e_S(o) &\propto \xi_{A_i}(a_i \mid \lambda_0) \xi_{A_j}(a_j \mid \lambda_0) \\
            &\propto 0\,.
        \end{split}
    \end{equation}
    Thus, necessarily $e_S(o) = 0$ for inconsistent outcomes.
    For consistent outcomes we can rigorously follow the proof in~\cite[Theorem 8.1]{abramsky2011SheaftheoreticStructureNonlocality}. Let us define the global distribution on an outcome $o \in O_X$, then:
    \begin{equation}
        d(o) \coloneq \sum_{\lambda_0 \in \Lambda} \mu(\lambda_0) \prod_{A \in X} \xi_{A}(a \mid \lambda_0)
    \end{equation}
    We naturally have that $d(o) \geq 0$, and it normalizes to one:
    \begin{equation}
        \begin{split}
            \sum_{o\in O_X} d(o) &= \sum_{o\in O_X}\sum_{\lambda_0 \in \Lambda} \mu(\lambda_0) \prod_{A \in X} \xi_{A}(a \mid \lambda_0) \\
            &= \sum_{\lambda_0 \in \Lambda} \mu(\lambda_0) \prod_{A \in X} \sum_{a} \xi_{A}(a \mid \lambda_0) \\
            &= \sum_{\lambda_0 \in \Lambda} \mu(\lambda_0) \\
            &= 1
        \end{split}
    \end{equation}
    We finally need to check that it marginalizes to the good value for a given context $S \in \mathcal{S}$:
    \begin{subequations}
	    \begin{align}
		\label{eq:marginal-proof-thm-NC} e_S(o) &= \sum_{\lambda_0} \mu(\lambda_0)  \left [ \prod_{j=1}^{N} \xi_{A_j}(a_j \mid \lambda_0) \right ] \\
		\label{eq:marginal-proof-thm-NC2} &= \sum_{\lambda_0} \mu(\lambda_0)  \prod_{A \in \bar{S}} \xi_{A}(a \mid \lambda_0) \\
		&= \restr{d}{\bar{S}}(o_{\bar{S}})
	    .\end{align}
    \end{subequations}
    Where we went from Equation~\eqref{eq:marginal-proof-thm-NC} to Equation~\eqref{eq:marginal-proof-thm-NC2} by using the fact that $\xi_{A}^2(a \mid \lambda_0) = \xi_{A}(a \mid \lambda_0)$ for all $A \in X$, thus ending with non-repeating measurements only.

    The proof for the converse statement, i.e., the existence of a hidden variable model given a global probability distribution, follows naturally from~\cite[Prop. 3.1]{abramsky2011SheaftheoreticStructureNonlocality}, where now for all $S \in \mathcal{S}$ we set
    \begin{equation}
        e_S(o) =  \left\{\begin{matrix*}[l]
 \restr{d}{\bar{S}}(o_{\bar{S}}),\, &\textup{if } o \textup{ is consistent}\\
 0,\,  &\textup{otherwise}
\end{matrix*}\right.
    \end{equation}
    Hence, $e_S(o)$ is explained by a factorizable hidden variable model. Then one can pick any transfer function that induces no-disturbance, such as Dirac delta functions, i.e. $\Gamma(\lambda_1  \mid \lambda_0, a) = \delta_{\lambda_1, \lambda_0}$. As we allow for repeating instruments, one can use the property $\xi_A(a\mid \lambda) = \left [\xi_A(a\mid \lambda)\right ]^n$ which holds for any $n\in \mathbb{N}^+$ when the response functions are deterministic.
\end{proof}

\subsection{Proof of Theorem~\ref{thm:ND-and-OI-are-NC}} \label{proof:ND-and-OI-are-NC}

We start by showing that the notion of ND simplifies with the assumption of OI:
\begin{lemma}[ND with the assumption of OI] \label{lemma:ND-assumption-OI}
    Let and OI ND HVM $\HVM$, then for any $A \nd B$ the following holds:
    \begin{equation}
    \begin{split}
        &\forall \lambda_0\colon \forall a \in O_A\colon \forall b \in O_B\colon\\
        &\quad \sum_{\lambda_1} \Gamma_{A}(\lambda_1 \mid \lambda_0) \xi_{B}(b \mid \lambda_1) = \xi_{B}(b \mid \lambda_0)
    \end{split}
    \end{equation}
\end{lemma}

\begin{proof}
    This comes quite straightforwardly from Definition~\ref{def:no-disturbance}, where the first term $\xi_{A}(a\mid \lambda_0)$ is ruled out by summing over $a$ since the transfer function is no more outcome dependent.
\end{proof}

We now turn to the proof of Theorem~\ref{thm:ND-and-OI-are-NC}.

\begin{proof}
    Let an OI ND HVM $\HVM$ on a sequential scenario $\XMOSEQ$, then we will show that this takes the form of a factorizable HVM. Let $S \in \mathcal{S}$ and two instruments $A_1, B_2 \in S$, by definition, we have $A_1 \nd B_2$ and:
    \begin{subequations}
	    \begin{align}
		    \label{eq:proof-OD-ND-eq-1} p(a_1, a_2 \mid A_1, A_2) &= \sum_{\lambda_0, \lambda_1} \mu(\lambda_0) \xi_{A_1}(a_1 \mid \lambda_0) \Gamma_{A_1}(\lambda_1\mid \lambda_0, a_1) \xi_{A_2}(a_2\mid \lambda_1) \\
		    \label{eq:proof-OD-ND-eq-2}&=\sum_{\lambda_0} \mu(\lambda_0) \xi_{A_1}(a_1 \mid \lambda_0) \xi_{A_2}(a_2 \mid \lambda_0)
	    .\end{align}
    \end{subequations}
    Where we went from Equation~\eqref{eq:proof-OD-ND-eq-1} to Equation~\eqref{eq:proof-OD-ND-eq-2} by using Lemma~\ref{lemma:ND-assumption-OI}.
    Because this is a factorizable HVM, we can follow the same reasoning as applied in the proof of Theorem~\ref{thm:ND-and-OD-are-NC}, in Appendix~\ref{proof:ND-and-OD-are-NC}.
\end{proof}

\end{document}